\documentclass[11pt]{article}
\usepackage[margin=1in]{geometry}

\usepackage{amsmath}
\usepackage{amsthm}
\usepackage{bbm}
\usepackage{amsfonts}
\usepackage{xcolor}
\usepackage{bbm}
\usepackage{cancel}
\usepackage{hyperref}
\usepackage{mathtools}

\usepackage{thmtools} 
\usepackage{thm-restate}
\usepackage{natbib}

\DeclareMathOperator{\E}{\mathbb{E}}%
\DeclareMathOperator{\R}{\mathbb{R}}%

\newcommand{\D}{\mathcal{D}}

\newcommand{\vect}[1]{\mathbf{#1}}
\newcommand{\expvi}{\E_{k \sim p} v_k^i}
\newcommand{\expv}{\E_{k \sim p} V_k}

\newcommand{\alternativeset}{\mathbb{Y}}

\newcommand{\sigmaj}[1]{\vec{\sigma}^{\,{#1}}}
\newcommand{\sigmaone}{\vec{\sigma}^{\,\boldsymbol{\cdot}, 1}}
\newcommand{\sigmatwo}{\vec{\sigma}^{\,\boldsymbol{\cdot}, 2}}

\newtheorem{lemma}{Lemma}
\newtheorem*{lemma*}{Lemma}
\newtheorem{corollary}{Corollary}
\newtheorem{theorem}{Theorem}
\newtheorem*{theorem*}{Theorem}
\newtheorem{proposition}{Proposition}
\newtheorem{observation}{Observation}
\theoremstyle{definition}\newtheorem{definition}{Definition}
\theoremstyle{definition}\newtheorem{remark}{Remark}
\theoremstyle{definition}

\begin{document}

 \title{Public Projects with Preferences and Predictions}
 \author{Mary Monroe and Bo Waggoner \\
 University of Colorado Boulder}
 \date{}
 
 \maketitle

\begin{abstract}
    When making a decision as a group, there are two primary paradigms: aggregating preferences (e.g. voting, mechanism design) and aggregating information (e.g. discussion, consulting, forecasting).
    Almost all formally-studied group decisionmaking mechanisms fall under one paradigm or the other, but not both.
    We consider a public projects problem with the objective of maximizing utilitarian social welfare.
    Decisionmakers have both preferences, modeled as utility functions over the alternatives; and information, modeled as Bayesian signals relevant to the alternatives' external welfare impact.
    Aligning incentives is highly challenging because, on the one hand, agents can provide bad information in order to manipulate the mechanism into satisfying their preferences; and on the other hand, they can misreport their preferences to favor selection of an alternative for which their information rewards are high.

    We propose a two-stage mechanism for this problem.
    The forecasting stage aggregates information using either a wagering mechanism or a prediction market (the mechanism is modular and compatible with both).
    The voting stage aggregates preferences, together with the forecasts from the previous stage, and selects an alternative by leveraging the recently-studied \emph{Quadratic Transfers Mechanism}. 
    We show that, when carefully combined, the entire two-stage mechanism is robust to manipulation of all forms, and under weak assumptions, satisfies Price of Anarchy guarantees.
    In the case of two alternatives, the Price of Anarchy tends to $1$ as natural measures of the ``size'' of the population grow large.
    In most cases, the mechanisms achieve a balanced budget as well.
    We also give the first nonasymptotic Price of Anarchy guarantee for the Quadratic Transfers Mechanism, a result of independent interest.
\end{abstract}

\section{Introduction} \label{sec:intro}
We study the public projects problem in which a group of decisionmakers is faced with selecting one of multiple alternative projects to undertake. 
Classically, the objective is to aggregate the \emph{preferences} of the decisionmakers in order to make a choice which most benefits the group as a whole.
Here, preferences will be modeled as quasilinear utility functions in a standard welfare-maximization framework, where a mechanism such as VCG might be employed.

Instead of preferences, a different philosophy for group decisionmaking is to focus on aggregating \emph{information} about which alternative is the best.
For example, a committee could make a decision either by voting (preferences), or by discussing until consensus is reached (information).
In the information-aggregation paradigm, \emph{decision markets}~\citep{othman2010decision} have been proposed as a group decisionmaking mechanism.
There, forecasts about the impact of each alternative are produced by running prediction markets -- financial markets whose prices reflect predictions -- and then a decision is automatically chosen based on the predicted-best alternative.
We likewise focus on \emph{predictions} as a source of credible, verifiable information.

A public projects problem can generally benefit from aggregation of both preferences and information.
Can we design a formal mechanism for groups to make decisions based on both?
There are only a few formal proposals in this spirit (see Section \ref{sec:related-work}).
When agents can serve as both information aggregators and preference-driven decisionmakers, they may have incentives to manipulate in both roles.

\paragraph{Model and main result.}
In this paper, we will propose a mechanism that aggregates both preferences and predictions to select a public project.
We consider monetary mechanisms and assume quasilinear utility, seeking social welfare guarantees.
In our setting, each participant $i$ has a value $v_k^i \geq 0$ for each project $k$, with $V_k = \sum_i v_k^i$.
Meanwhile, project $k$ has some \emph{external welfare impact} $B_k$ representing the impact on the social welfare of non-decisionmakers.
For example, $B_k$ could reflect an externality imposed by project $k$ such as carbon emissions or another environmental impact (see Section \ref{subsec:outline}).
In addition to their private values, the participants have private information, modeled as signals, relevant to each $B_k$.
We consider the \emph{expected} external impact of alternative $k$ conditioned on all signals as $B_k^*$.
We seek to aggregate the preferences $\{v_k^i\}$ and expected external impacts $\{B_k^*\}$ in order to maximize total social welfare: $\arg\max_k (V_k + B_k^*)$, the sum of welfare of the participants and the expected external impact.

The \emph{Synthetic players QUAdratic transfer mechanism with Predictions (SQUAP)} begins with an information aggregation phase, which may be a wagering mechanism or a prediction market, to elicit $B_k^*$ for each $k$.
Then, it uses a ``voting'' stage, based on the recently-introduced \emph{Quadratic Transfers Mechanism (QTM)} discussed below, to elicit preferences and select an alternative.
We obtain the following result on the \emph{Price of Anarchy} (the worst-case ratio between social welfare achieved by (1) an equilibrium of the mechanism and (2) the best alternative):
\begin{theorem*}[Main result, informal]
	In SQUAP, for the correct choice of parameters, for $m=2$ alternatives,
	\[ \text{Price of Anarchy} \geq 1 - O(T^{-2/5}) , \]
	where the \emph{spread} $T = \frac{\text{total welfare}}{\text{maximum individual value}}$ measures the size of the game.
\end{theorem*}
More formally,  $T = \frac{\max_k (V_k + B_k^*)}{\max_{i,k} v_k^i}$, the ratio of first-best welfare to the largest individual value.
Thus, as the total welfare grows large relative to the maximum influence of any one person's preference, the Price of Anarchy tends to one.
For the case of $m \geq 3$ alternatives, we also obtain general results for SQUAP as a function of the Quadratic Transfer Mechanism's Price of Anarchy, which is an open problem to fully analyze.
\citet{Eguia23} obtain a Price of Anarchy asymptotically tending to $1$ for QTM for $m \geq 3$ under some assumptions, which implies via our results the same guarantee for SQUAP.

SQUAP does require assumptions on the mechanism's knowledge in order to implement it.
We give several suggestions for relaxing these assumptions; see Section \ref{subsec:outline}.
More broadly, the result proves that aligning incentives to base a decision on both preferences and predictions is possible, even when strategic participants act as both forecasters and decisionmakers.

\subsection{Outline of Techniques} \label{subsec:outline}

\paragraph{The QTM.}
Recent work~\citep{Eguia19,Eguia23} has introduced and studied the \emph{Quadratic Transfers Mechanism (QTM)} for public projects.
Inspired by quadratic voting~\citep{Lalley2018,Goeree2017}, the mechanism collects bids (``votes'') on each alternative, with cost rising quadratically, and picks an alternative according to a probabilistic ``softmax'' of the vote totals.
By redistributing payments, the mechanism is budget-balanced.

QTM turns out to have several nice properties for our two-stage problem that will be discussed below.
Our first technical result is an analysis of the QTM including Price of Anarchy guarantees (Section \ref{sec:qtm}).
Our approach builds on the work of \citet{Eguia19}, which gives asymptotic welfare bounds in a model with agents drawn i.i.d. from bounded value distributions, where realized values are common knowledge.
We consider any worst-case set of agents with no bounds or distributional assumptions, and we analyze pure-strategy Nash equilibrium.

\paragraph{Incorporating information on external welfare.}
In Section \ref{sec:external-welfare}, we take one step toward incorporating information into the QTM.
Before discussing these results, we digress briefly to explain the model of information in this paper.

Broadly, information can be relevant to the decisionmakers' own utilities and/or to an ``external'' welfare impact on non-decisionmakers.
Consider shareholders of a non-profit corporation voting on a proposal that includes a significant climate impact.
Shareholders' utilities may incorporate both the revenue generated and the relevance of the climate impact to themselves.
But in addition, the climate impact has an \emph{externality} on non-decisionmakers.
The charter of the non-profit could include an obligation to incorporate such externalities into decisions.
The total welfare of alternative $k$ is therefore $V_k + B_k^*$, the welfare of the decisionmakers plus the expected external welfare impact.

We consider a standard Bayesian model of information where decisionmakers observe private signals (the full model is presented in Section \ref{sec:general-mechanism}).
For simplicity in this paper, we consider information that is \emph{only relevant to $B_k$}, the external welfare impact.
We assume that decisionmakers fully know their own values before receiving any signals.
This model still captures the inherent strategic tension between satisfying preferences and acting on good information.
Further, our mechanism can easily be applied to a setting where decisionmakers' own preferences depend on private information, such as in common-value settings.
However, analyzing the strategic properties in such cases is left as an open problem.

Coming back to the results, in Section \ref{sec:external-welfare}, we modify the QTM to incorporate external welfare impacts in a case where the expectations $\{B_k^*\}$ of these impacts are known public information and non-manipulable.
We show that, by simulating synthetic agents, we can retain QTM's welfare guarantees as measured against total welfare $\max_k (V_k + B_k^*)$.
Because the mechanism simulates participants in the mechanism, it requires some information about values in order to play correctly, but this information is relatively minimal.

\paragraph{Combined preference-prediction mechanisms.}
In Section \ref{sec:predictions} we define a two stage mechanism, \emph{Synthetic players QUAdratic transfer mechanism with Predictions (SQUAP)}.
The first stage has agents participate in an information-aggregation mechanism, such as a wagering mechanism or prediction market, to produce estimates $\{\hat{B}_k\}$ of the external welfare impacts of the projects.
To produce the final decision, the second stage runs the Synthetic Players QTM from Section \ref{sec:external-welfare} with the aggregated estimations from the first stage as input.

To prove a Price of Anarchy bound, we take the following steps.
First, we consider two commonly-studied mechanisms for eliciting predictions: wagering mechanisms and prediction markets.
In each, the penalty for misreporting predictions can scale unboundedly as the mechanism grows ``large,'' even for misreports with small impact on the welfare estimates.
This ensures that decisionmakers, if their utilities are small relative to total social welfare, cannot gain more than $\epsilon$ by manipulating the final predictions, and even when they do, the final predictions are still highly accurate.
Second, we use an importance-weighting technique from the \emph{decision markets} literature~\citep{Chen11} so that forecasters' expected payoffs are independent of the decision of the mechanism.
This technique ensures that manipulating votes (preferences) cannot influence the expected rewards from providing information.
This property relies on the QTM because the mechanism always produces positive probabilities for each alternative.
Under some mild assumptions, we provide a Price of Anarchy guarantee, our main result, based both on the spread of the values $T$ and the size of the information-aggregation mechanism.
The latter can be tuned to match $T$, and it can be subsidized (to an extent) by the revenue from the QTM.

One point to emphasize is that both stages of our mechanism are non-truthful in equilibrium.
However, we show that the information stage is arbitrarily close to truthful and forecasts are highly accurate, while the QTM's welfare properties hold in non-truthful equilibrium.

\subsection{Related Work} \label{sec:related-work}
\paragraph{Public projects.}
There is naturally a significant amount of work on public projects from a mechanism-design perspective.
We focus on the setting of this paper, which is social welfare maximization with quasilinear-utility agents.
A standard solution is to use the Vickrey-Clarke-Groves mechanism (VCG)~\citep{vickrey1961counterspeculation,clarke1971multipart,groves1973incentives} which has an equilibrium that is truthful and maximizes social welfare.
However, VCG has some undesirable properties, including lack of a Price of Anarchy bound (there exist arbitrarily bad equilibria) and uncertain revenue (it may be very large or zero).
Literature addressing the problem with more nuance than our basic model considers the \emph{excludable} and \emph{non-excludable} cases, e.g. \citet{ohseto2000characterizations}.
The \emph{smoothness framework}~\citep{roughgarden2017price} has been used to analyze Price of Anarchy of public projects, particularly in a combinatorial setting.
In particular, \citet{lucier2013equilibrium} show that a first-price mechanism can achieve a Price of Anarchy of $1$ (i.e. guarantee optimal welfare) in a sequential setting where bidders iteratively submit public bids.

\paragraph{QTM.}
Recent research has considered transfer-based mechanisms which use quadratic payments ~\citep{Weyl2013OriginalQV,Lalley2018,Goeree2017}.
\citet{Lalley2018,Lalley2019} propose a binary-outcome quadratic voting mechanism and show it is asymptotically efficient.
\citet{Weyl2017robustness} then study this mechanism's sensitivity to collusion, fraud, and voter mistakes when the number of agents is large.
A number of other works study the behavior or empirical performance of quadratic voting~\citep{Chandar2019quadratic,Casella2019,Quarfoot2017,Goeree2017}.

Building on \citet{Lalley2019}, the QTM in particular is proposed by \citet{Eguia19} for $m \geq 2$ alternatives.
The authors consider a setting where participants' values for the $m$ alternatives are drawn i.i.d. from a distribution over $[0, 1]^m$, and realized values are common knowledge.
They prove that for an undetermined choice of the mechanism's scale parameter $c$, equilibria exist, and show that the probability of selecting the highest-welfare alternative approaches 1 as the number of agents diverges.
\citet{Eguia23} then study the same mechanism in a similar setting but under Bayes-Nash mixed-strategy equilibrium, with similar results.
The part of our work on the QTM builds heavily on the approaches in \citet{Eguia19,Eguia23}, but our results address the nonasymptotic regime without distributional assumptions.
This includes concrete bounds with small constants.

\paragraph{Information and decisionmaking.}

We first note that some classic problems can be viewed as melding information with decisionmaking.
In voting settings, jury theorems and related work (see e.g. \citet{conitzer2005common}) interpret voting rules as aggregating information.
A common-value auction (see e.g. \citet{kremer2002information}) is a case of private information relevant to utilities and decisionmaking as well.

A number of works study mechanisms for eliciting information from experts and using it to make a decision.
The most relevant mechanisms for this paper are decision markets~\citep{hanson1999decision,othman2010decision}, which simultaneously operate a prediction market for each alternative $k$.
(A prediction market is a financial market designed to aggregate beliefs into a consensus forecast, reflected in the prices of the financial products being traded~\citep{hanson2003combinatorial}.)
The decision market mechanism selects the alternative that is best according to the aggregated predictions (or uses a similar rule).
Even without preferences over the alternatives, decision markets can have complex incentive misalignment problems, although these can be fixed by the importance-weighting technique of \citet{Chen11}.

\paragraph{Consulting with experts.}
There are many works in which a single decisionmaker elicits information from one or more experts prior to making a decision, and the experts wish to influence the decision.
This includes the ``cheap talk'' model of \citet{crawford1982strategic} as well as Bayesian Persuasion~\citep{kamenica2011bayesian}.
Somewhat closer to our setting, \citet{oesterheld2020decision} consider ``decision scoring rules'' that incentivize truthful predictions \emph{and recommendations} by an expert.

For aggregating information from a group, e.g. \citet{Gerardi2009} consider a Bayesian setting in which a decisionmaker elicits information from a number of experts and aggregates it to make a decision.
To incentivize truthfulness, the decisionmaker with small probability audits an expert's report against the others (similar to peer prediction~\citep{miller2005eliciting}) and picks an alternative desired by the audited agent.
Such settings typically do not face the concern in this paper, where the decisionmakers are a large public group and may overlap with the forecasters.

\paragraph{Decisions from preferences and predictions.}
A few works are much closer to the spirit of our motivation, if technically quite different.
They generally analyze voter preferences as a function of the information to be aggregated.
In contrast, we model preferences as fixed with information as an orthogonal axis.
In the ``Wisdom-of-the-Crowd Voting Mechanism'' of \citet{schoenebeck2021wisdom}, each participant is both a voter and a holder of private information.
The mechanism elicits both information and a preference in a single shot, then aggregates both to make a binary decision.
The authors show that, with high probability, the mechanism selects despite strategic behavior the ``majority wish'' alternative that more voters would prefer if fully informed.

\citet{amanatidis2022decentralized} consider a similar problem motivated by blockchain applications.
Voters have private information and participate in an approval vote among $m \geq 2$ alternatives.
As with \citet{schoenebeck2021wisdom}, information aggregation takes place within the mechanism rather than through an explicit phase prior to voting.
The authors prove that the mechanism can achieve a Price of Anarchy of $\frac{1}{2}$.

In \citet{jackson2013deliberation}, a majority-vote between two alternatives takes place after a ``deliberation'' stage.
The experts who take part in deliberation cannot also be voters (unlike in our model), and cannot misreport arbitrarily, but only choose to either reveal or hide their information in order to influence the voters.
Other works in this spirit are \citet{alonso2016persuading,schnakenberg2015expert}.

\section{General Setting and Mechanism} \label{sec:general-mechanism}
We consider the problem of \emph{public projects with predictions}, in which a group of $n$ agents needs to select one of $m \geq 2$ alternatives.

We will use $i,j$ to denote generic agents and $k,l$ to denote generic alternatives.
Let $v_k^i \in \R_{\geq 0}$ be agent $i$'s value for alternative $k$.
We assume that all values are nonnegative.
We consider a Nash equilibrium setting where values are common knowledge among the agents.
Denote $\vect{v}^i = ( v_1^i, v_2^i, \ldots, v_m^i)$ as the vector of agent $i$'s values across all $m$ alternatives.
The \emph{aggregate value} for alternative $k$ is $V_k := \sum_{i=1}^n v_k^i$.

Each agent also receives a signal $S_i$ before the mechanism begins.
Let random variable $B_k \in \R_{\geq 0}$ represent the impact to \emph{external welfare} conditional on the $k$th alternative being chosen.
Signals are drawn jointly from a common-knowledge prior $\D$ jointly along with the welfare impacts, i.e. we have $(S_1,\dots,S_n,B_1,\dots,B_m) \sim \D$.
We let $\D_k$ be the marginal distribution of $\D$ over alternative $B_k$.
The mechanism does not know $\D$.

The expected external welfare impact of alternative $k$, given all available information, is $B_k^* := \E[B_k \mid S_1,\dots,S_n]$.
The expected social welfare of alternative $k$ is $W_k := V_k + B_k^*$, a random variable depending on the signals.
The objective is to maximize expected social welfare (referred to as simply social welfare).
We assume WLOG that $W_1 \geq \cdots \geq W_m$, so that choosing alternative $1$ maximizes social welfare.

\paragraph{Our mechanism: framework.}
We will consider mechanisms with two stages:

\begin{enumerate}
	\item{\emph{Forecasting.}} In the first stage, agents participate in an aggregation mechanism.
        We will formalize the mechanisms we use in Section \ref{sec:predictions}.
        The aggregation mechanism produces a vector of aggregated forecasts, the \emph{external welfare estimates} $\vect{\hat{B}} = (\hat{B}_1,\dots,\hat{B}_m)$.
	The aggregation mechanisms we consider, prediction markets and wagering mechanisms, reward participants based on the eventual observation of the actual outcome\footnote{In reality, the welfare $B_k$ may not be directly observable, but we may use estimates or proxies for $B_k$ instead and modify the mechanism accordingly.} $B_k$ once an alternative $k$ is chosen.
        At that time, the mechanism collects a payment $\pi^{i,1}$ from each agent $i$ as a function of the reports to the mechanism, the eventually-selected alternative $k$, and the realized outcome of $B_k$.
	\item{\emph{Decisionmaking.}} In the second stage, agents participate in the decisionmaking mechanism.
        For this stage, we will use a variant of the Quadratic Transfers Mechanism, defined in Section \ref{sec:qtm}.
	The decisionmaking mechanism, which uses the external welfare estimates $\{\hat{B}_k\}$ as part of its process, outputs a probability distribution $\vect{p} \in \Delta_m$.
        It also collects a payment $\pi^{i,2}$ from each agent $i$ as a function of the reports to the mechanism.
        The final group decision is $k$ drawn randomly according to $\vect{p}$.
\end{enumerate}
If the joint mechanism selects alternative $k$ under $\vect{p}$, the \emph{social welfare} of the mechanism is $\E_{k \sim \vect{p}} W_k := \E_{k \sim \vect{p}}[V_k + B_k^*]$.
Each agent's ex-post utility when alternative $k$ is selected is $v^i_k - \pi^{i,1} - \pi^{i,2}$.

\paragraph{Approach: backward induction on the stages.}
As discussed, there are conceptual incentive challenges in designing a mechanism for this context.
By using the two-stage approach, we modularize the problem to an extent.
In Section \ref{sec:predictions}, we will be able to use properties of wagering mechanisms and prediction markets to bound the impact of manipulation in stage 1.
However, before considering stage 1, we need to solve the second stage of the game.
Thus, it is useful for us to begin with analyzing the Quadratic Transfers Mechanism in Section \ref{sec:qtm}.
Then we will extend it in Section \ref{sec:external-welfare} to incorporate external welfare estimates (from non-strategic sources).

\section{Quadratic Transfers Mechanism} \label{sec:qtm}
In this section, we give results, including a Price of Anarchy bound, for the Quadratic Transfers Mechanism (QTM) for the public projects problem.
These results complement existing analyses~\citep{Eguia19,Eguia23}, and will also serve as a foundation for our results in later sections combining predictions and preferences.

\subsection{Model}
In this section, we consider the special case of our setting without external welfare; that is, the classic public projects setting.
There are $n$ agents with preference profile $\vect{v}^i$ who wish to select one of $m \geq 2$ alternatives.
In this section, the welfare of alternative $k$ is $V_k = \sum_{i=1}^n v_k^i$.
Without loss of generality, in this section we number the alternatives such that $V_1 \geq V_2 \geq \ldots \geq V_m$.
This is a special case of the setting in Section \ref{sec:general-mechanism}, where it is known that $B_k = 0$ for all alternatives $k$ and agents only participate in the second stage. 

A decisionmaking mechanism takes the actions of agents as input and outputs a probability distribution $\vect{p} \in \Delta_m$, along with a net payment from each agent.
We assume quasilinear utility: if the mechanism's output is $\vect{p}$ and if agent $i$ makes a net payment $\pi^i$, then $i$'s \emph{expected utility} is given by $\sum_k p_k v_k^i - \pi^i$.
The social welfare of the mechanism is $\sum_k p_k V_k$, the expected aggregate value produced.
The optimal or ``first-best'' social welfare is $V_1$, the aggregate value of the best alternative.

In general, the Price of Anarchy of a mechanism is the worst case ratio between the social welfare of the mechanism in equilibrium and the optimal social welfare.
Formally, given a mechanism and a value profile $\vect{v} = (\vect{v}^1,\dots,\vect{v}^n)$, let $\text{pNE}(\vect{v})$ denote the set of distributions $\vect{p} \in \Delta_m$ induced by the mechanism in any pure-strategy Nash equilibrium.
Then the \emph{pure-strategy Price of Anarchy} of the game with respect to $\vect{v}$ is
\begin{equation}
	\text{pPoA}{(\vect{v})} = \frac{\min_{\vect{p} \in \text{pNE}(\vect{v})} \sum_k p_k V_k} {V_1}. \label{eq:poa}
\end{equation}
We define the Price of Anarchy relative to a value profile $\vect{v}$ because we will give bounds depending on properties of $\vect{v}$; for instance, a Price of Anarchy bound that improves if all agents are ``small''.
However, for convenience, we may refer to the pPoA without explicitly including $\vect{v}$.

We say a mechanism achieves \emph{budget balance} if the net total payment of all agents is zero.

\paragraph{The QTM.}
In the Quadratic Transfers Mechanism (QTM)~\citep{Eguia19,Eguia23}, each agent $i$ simultaneously and privately submits a vector $\vect{a}^i \in \R^m$, where the $k$th entry $a_k^i$ represents a number of ``votes'' for alternative $k$.
Votes are allowed to be negative, representing voting against that alternative.
We let $A_k := \sum_i a_k^i$ denote the \emph{aggregate votes} for alternative $k$.

Each agent is charged $c \sum_k (a_k^i)^2$, i.e. the sum of squares of their votes multiplied by a parameter $c > 0$ chosen in advance by the mechanism designer.
Then, the mechanism chooses one alternative according to the ``softmax'' distribution $\vect{p} \in \Delta_m$ defined by
  \[ p_k = \frac{e^{A_k}}{\sum_l e^{A_l}}. \]
The QTM redistributes each agent's payment equally among the other voters, so that $i$ receives
\begin{equation} \label{eq:redistr-term}
	\frac{c}{n-1} \sum_{j \neq i} \sum_k (a_k^j)^2.
\end{equation}
Thus, the QTM is budget-balanced.\footnote{This redistribution does not affect the incentives of each voter, so it can be modified without changing the strategic properties of the mechanism. We will utilize this in some of proposed mechanisms later in the paper.}
Agent $i$'s utility function $u^i$, taken in expectation over the mechanism's randomness, as a function of the votes $\vect{a} \coloneqq (\vect{a}^1, \vect{a}^2, \ldots, \vect{a}^n)$, is
\begin{equation} \label{eq:utility}
	u^i(\vect{a}) = \sum_{k=1}^m p_k v_k^i - c \sum_{k=1}^m (a_k^i)^2 + \frac{c}{n-1} \sum\limits_{\substack{j=1 \\ j \neq i}}^n \sum_{k=1}^m (a_k^j)^2 .
\end{equation}
We note that agent $i$'s actions only affect the first two terms, so for strategic analysis, the third (redistribution) term is irrelevant.

\subsection{Price of Anarchy: Two Alternatives} \label{sec:two-alternatives}

In this section, we present our results for the QTM on $m=2$ alternatives. 
Conceptually similar results have already been shown by \citet{Eguia19,Eguia23} (see Section \ref{sec:related-work}), but they focus more on the challenges of the $m \geq 3$ setting and have different distributional and model assumptions than we do here.
We note that this section relies on an analysis of equilibria provided in Appendix \ref{app:equilibrium}.
Specifically, we first show pure Nash equilibria exist given a lower bound on $c$, and prove that for $m = 2$ alternatives equilibria are unique:

\begin{proposition} \label{prop:equilibrium-existence-uniqueness}
	In the QTM on agent value profile $\vect{v}$, if the mechanism chooses $c \geq \tfrac{1}{2} \max_{i,k} v_k^i$, then a pure-strategy Nash equilibrium exists.
        If $m=2$, it is unique.
\end{proposition}
The result is proven as Propositions \ref{prop:equilibrium-existence} and \ref{prop:equilibrium-uniqueness} in Appendix \ref{app:qtm}.

Then we observe several properties of the first order conditions for agent votes.
We also characterize the revenue (before redistributions) in terms of ``disagreement'' in the type profile in Appendix \ref{app:two-alternatives}, and briefly touch on the case where there are more than two alternatives in Appendix \ref{appendix:qtm-m-alternatives}.

\paragraph{Price of Anarchy result.}
Recall that $V_1$ is the aggregate value of the better alternative, i.e. optimal welfare, and $V_2$ is the aggregate value of the other.
Our main result, Theorem \ref{thm:qtm-main-thm}, is that the QTM's pure-strategy Price of Anarchy (pPoA) is always at least $\tfrac{1}{2}$ and approaches $1$ rapidly as the \emph{spread} of the game --- the ratio of $V_1$ to the largest individual value --- grows.
We also show in Corollary \ref{cor:qtm-gap-poa} a faster rate of convergence to $1$ when the \emph{gap} --- the difference between $V_1$ and $V_2$ --- grows.
Finally, we will characterize the total payments of the QTM in terms of the amount of ``disagreement'' of the participants.

\begin{theorem} \label{thm:qtm-main-thm}
	Define the \emph{spread} of the game to be $T \coloneqq \frac{V_1}{max_{i, k} v_k^i}$.
	In the two-alternative case, there exists a choice of $c$ such that the QTM achieves a Price of Anarchy
		\[ \text{pPoA}(\vect{v}) \geq \max \left\{ \frac{1}{2} ~,~~ 1 - \left(\frac{2}{T}\right)^{2/5} \right\} .  \]
\end{theorem}
The result is proven in Appendix \ref{app:two-alternatives}.
The proof relies on a careful analysis of first-order conditions.
In particular, if the spread of the game is large, then there are two cases: $V_2$ is also large, in which case the welfare of any mechanism is large; or else the \emph{gap} $V_1-V_2$ is large, in which case we can prove a pPoA bound.
It also yields:
\begin{corollary} \label{cor:qtm-gap-poa}
	Define the \emph{gap} of the game to be $G \coloneqq \frac{V_1 - V_2}{max_{i, k} v_k^i}$.
	In the two-alternative case, there exists a choice of $c$ such that the QTM achieves a Price of Anarchy
		\[ \text{pPoA}(\vect{v}) \geq p_1 \geq \max \left\{ \frac{1}{2} ~,~~ 1 - \left(\frac{4}{G}\right)^{2/3} \right\} .  \]
\end{corollary}

The case of $m \geq 3$ is highly complex.
\citet{Eguia19,Eguia23} give an asymptotic analysis for a Bayesian setting showing a Price of Anarchy approaching one.
For a Nash equilibrium setting with no assumptions, we do not obtain such a result, but we do have a Price of Anarchy bound that does not depend at all on the number of participants or their values.

\begin{proposition} \label{prop:m-poa}
	For the QTM on $m$ alternatives, for any choice of $c \geq \tfrac{1}{2}\max_{i,k} v_k^i$, we have $\text{pPoA} \geq \frac{1}{m}$.
\end{proposition}
The proof appears in Appendix \ref{appendix:qtm-m-alternatives}.
It involves iteratively dividing the alternatives into a group with ``high'' welfare (higher than the mechanism's expectation) and the remainder, an idea from \citet{Eguia23} used in a different way.

\section{Incorporating External Welfare} \label{sec:external-welfare}

In this section, we propose mechanisms for public projects with \emph{external welfare}, meaning estimates or projections of an externality of the project on non-decisionmakers.
We will suppose that the external welfare impacts are fixed and common knowledge of the agents and the mechanism.
Thus we are in a special case of the setting in Section \ref{sec:general-mechanism} without an aggregation stage, where the external welfare impacts are fully known in advance, i.e. $B_k = \hat{B}_k = B_k^*$ is the external welfare impact of alternative $k$.
In Section \ref{sec:predictions}, we will extend to the general case where external welfare impacts are forecasts from (manipulable) aggregation mechanisms.

\paragraph{Model.}
We define the problem of \emph{Public Projects with External Welfare} as an extension of the standard public projects setting where selecting alternative $k$ causes a known, fixed \emph{external welfare impact} $B_k$.
The $n$ agents continue to have values $\{v_k^i\}$, and the total value of alternative $k$ continues to be $V_k = \sum_{i=1}^n v_k^i$.
The welfare if $k$ is chosen is $W_k \coloneqq V_k + B_k$, and we assume without loss of generality that alternatives are numbered according to $W_1 \geq W_2 \geq \ldots \geq W_m$.
We let $\text{pPoA}(\vect{v}, \vect{B})$ be the pure-strategy Price of Anarchy of our mechanism in the setting $(\vect{v}, \vect{B})$, i.e. the worst-case ratio of welfare in pure-strategy Nash equilibrium to $W_1$.

\paragraph{Mechanism.}
We propose the \emph{Synthetic Players Quadratic Transfer Mechanism} for public projects with external welfare.
The mechanism works by creating a set of synthetic players whose preferences reflect the external welfare, then playing on their behalf in the QTM alongside the real participants.
It proceeds as follows:
\begin{enumerate}
  \item The mechanism is given $B_1,\dots,B_m$.
  \item Set $c \geq \tfrac{1}{2} \max_{i\leq n,k} v_k^i$.
  \item Define $\hat{n} \geq \frac{\max_k B_k}{2c}$ ``synthetic players,'' each with a value profile of $\frac{1}{\hat{n}} (B_1,\dots,B_m)$.
  \item Run the QTM on real agents $1,\dots,n$ and synthetic agents $\hat{1},\dots,\hat{n}$, where the mechanism plays on behalf of the synthetic players. Do not redistribute the payments of the synthetic players.
\end{enumerate}
We obtain the following Price of Anarchy guarantee for the mechanism:

\begin{proposition} \label{prop:synthetic-fixed-ppoa-bb-m}
  Suppose the QTM with $m$ alternatives on profile $\vect{v}$ guarantee $\text{pPoA}{(\vect{v})} \geq C(\vect{v}).$
  Then the Synthetic Players QTM on $m$ alternatives, for all external welfare vectors $\vect{B}$, achieves budget balance and guarantees $\text{pPoA}{(\vect{v}, \vect{B})} \geq C(\vect{v}).$
\end{proposition}

\begin{proof}
    Budget balance follows immediately from budget balance of the QTM.
    For each synthetic player $j=\hat{1},\dots,\hat{n}$, let her value vector be $\vect{\hat{v}}^j = \tfrac{1}{\hat{n}} \vect{B}$.
    Observe $c \geq \tfrac{1}{2} \max_{j,k} \hat{v}_k^j$, and by construction $c \geq \tfrac{1}{2} \max_{i\leq n,k} v_k^i$, so that a pure-strategy Nash equilibrium is guaranteed to exist by Proposition \ref{prop:equilibrium-existence-uniqueness}.
    Moreover, the welfare of each alternative $k$ in the game is $W_k$.
    The result follows.
\end{proof}
An immediate corollary, using the Price of Anarchy results in Theorem \ref{thm:qtm-main-thm} and Corollary \ref{cor:qtm-gap-poa} for $m = 2$, is:

\begin{corollary} \label{cor:synthetic-fixed-ppoa-bb}
  For the Synthetic Players QTM on two alternatives, for any profile $\vect{v}$, there exists a choice of parameter $c$ such that, for all external welfare vectors $\vect{B}$, the mechanism achieves budget balance and
  \begin{enumerate}
    \item $\text{pPoA}(\vect{v},\vect{B}) \geq \max \{ \tfrac{1}{2} , 1 - (\tfrac{2}{T})^{2/5} \}$, where $T = \frac{W_1}{\max_{i,k} v_k^i}$ is the spread;
    \item $\text{pPoA}(\vect{v},\vect{B}) \geq p_1 \geq \max \{ \tfrac{1}{2}, 1 - \left(\tfrac{4}{G}\right)^{2/3} \}$, where  $G = \frac{W_1 - W_2}{max_{i, k} v_k^i}$ is the gap.\footnote{Note that the inequality in terms of $p_1$ holds since any guarantees of the QTM for $\vect{p}$ extend to the external welfare setting by the proof of Proposition \ref{prop:synthetic-fixed-ppoa-bb-m}.}
  \end{enumerate}
\end{corollary}

\paragraph{Practical concerns and mechanism variants.}
To implement the mechanism, one needs to compute an equilibrium strategy for the synthetic players.
Although this only requires knowledge of the vote totals of the other players -- not their individual strategies -- it may be impractical in many cases.
One mitigating factor is that, in the two-alternative case, equilibrium is unique (see Proposition \ref{prop:equilibrium-existence-uniqueness}), so equilibrium selection is not a concern.

We also suggest two potential variants of the mechanism that remove the requirement for the mechanism to compute an equilibrium.
The first is a sequential voting game similar to the sequential first-price public projects mechanism of \cite{lucier2013equilibrium}.
We ask the ``real'' participants to cast votes publicly one at a time, with the mechanism going last on behalf of the synthetic players.
The second is a two-stage version, in which all ``real'' participants cast their votes in stage one simultaneously, and then the mechanism adds the synthetic votes in stage two.
In both of these suggestions, the mechanism computes the synthetic players' total votes and the output of the mechanism using the first-order conditions.
In the two-alternative case, for example, given the real players' votes $\{a_k^i\}$, the mechanism solves the fixed-point problem
  \[  p_1 = \frac{e^{\left(\sum_i a_1^i\right) + \tfrac{p_1p_2}{2c}(B_1 - B_2)}}{e^{\left(\sum_i a_1^i\right) + \tfrac{p_1p_2}{2c}(B_1 - B_2)} + e^{\left(\sum_i a_2^i\right) + \tfrac{p_1p_2}{2c}(B_2 - B_1)}} . \]
In practice, we expect both variants to be practical.
In particular, in large games, we do not expect participants to be able to strategize significantly to influence the result of the mechanism.
However, strategic analysis of these variants is left to future work.

\section{Incorporating Predictions} \label{sec:predictions}
Armed with the analyses in Sections \ref{sec:qtm} and \ref{sec:external-welfare}, we now return to the setting of public projects with predictions.
As a reminder from Section \ref{sec:general-mechanism}, an aggregation mechanism is operated first in order to advise a decisionmaking mechanism, which operates second.
The first stage produces as output forecasts $\vect{\hat{B}} = (\hat{B}_1,\dots,\hat{B}_m)$.
The second stage is to run the QTM with external welfare impacts $\vect{\hat{B}}$, as defined in Section \ref{sec:external-welfare}, which outputs a probability distribution $\vect{p}$ over the alternatives and selects from that distribution.

\subsection{The Two-Stage Game}
We begin by explicitly defining strategies and utilities in the two-stage game.
We remind the reader that the external welfare impact conditional on choosing alternative $k$, $B_k$, is now a random variable.
Agent $i$ receives a signal $S_i$ drawn jointly with $(B_1,\dots,B_m)$ from a common-knowledge prior $\D$.
When an agent acts in the first stage, we refer to them as a \emph{forecaster}; in the second stage, as a \emph{decisionmaker}.

As a forecaster, agent $i$ submits a prediction vector $\vect{\hat{b}}^i$ according to her strategy $\sigma^{i,1} : S_i \to \R^m$.
Let $\vect{\hat{b}} = (\vect{\hat{b}}^1, \ldots, \vect{\hat{b}}^n)$ be the set of prediction vectors across forecasters.
She is paid according to some function $\pi^{i,1}(\vect{\hat{b}}^i, \vect{p}, k,  b_k)$ depending on the eventual observation of the actual outcome $B_k$ of the alternative chosen by the joint mechanism.
Thus agent $i$'s expected utility $u^{i,1}$ in this stage corresponds to her expected payment: $u^{i,1}(\vect{\hat{b}}^i, \vect{p}) = \sum_{k} p_k \E_{b_k \sim \D_k} \pi^{i,1}(\sigma^{i,1}, \vect{p}, k,  b_k)$.
As a decisionmaker, agent $i$ also submits a vote vector $\vect{a}^i$ to the decisionmaking mechanism according to strategy $\sigma^{i,2}: (S_i, {\vect{\hat{b}}}, \vect{\hat{B}}, \vect{a}^{-i}) \to \R^m$. 
She receives some payment $\pi^{i,2}(\vect{a})$, so her expected utility $u^{i,2}$ in this stage (under the assumption of quasilinear utilities) is $u^{i,2}(\vect{a}, \vect{p}) = \sum_k p_k v_k + \pi_2^i(\vect{a}).$

In the sequential two-stage game, then, agent $i$'s strategy corresponds to $\sigmaj{i} = (\sigma^{i,1}, \sigma^{i,2})$, and her expected utility is $u^i(\vect{\hat{b}}^i, \vect{a}, \vect{p}) = u^{i,1}(\vect{\hat{b}}^i, \vect{p}) + u^{i,2}(\vect{a}, \vect{p})$. 
Let $\sigmaone = (\sigma^{1, 1}, \ldots, \sigma^{n,1})$ denote the set of strategies in the first stage, $\sigmatwo = (\sigma^{1,2}, \ldots, \sigma^{n,2})$ the set of strategies in the second stage, and $\vec{\sigma} = (\sigmaj{1}, \ldots, \sigmaj{n})$ the strategy profile over the entire mechanism.

\paragraph{Mechanism overview.}
The specific mechanism implementation we propose, the \textbf{Synthetic players QUAdratic transfer mechanism with Predictions (SQUAP)}, is as follows: 
\begin{enumerate}
  \item Forecasters participate in the aggregation mechanism, which produces $\vect{\hat{B}} = (\hat{B}_1,\dots,\hat{B}_m)$.
  \item We run the Synthetic Players QTM defined in Section \ref{sec:external-welfare} using external welfare impacts $\vect{\hat{B}}$.
  \item The output distribution is $\vec{p}$ and the selected alternative is $k$.
  \item Later, $B_k$ is observed and the aggregation mechanism rewards participants based on $\vec{p}, k, B_k$.
\end{enumerate}
Redistribution of QTM payments has the potential, at least in theory, to cause incentive problems.
For theoretical analysis, we assume that the mechanism keeps the payments or ``burns'' them.
In practice, designers may choose to redistribute them or use them to subsidize the aggregation mechanism, as the incentive concern is likely low, but we leave that analysis to future work.

\subsection{SQUAP Price of Anarchy}

\paragraph{Incentive alignment.}
Perhaps surprisingly, we only require two natural properties of an aggregation mechanism in order to prove welfare guarantees in equilibrium of the combined mechanism.
These properties address the two possible incentive problems previously mentioned: a voter may manipulate the predictions in order to boost some alternatives' chances of being chosen, and a forecaster may manipulate the decisionmaking mechanism in order to increase the rewards for their predictions.
First is the incentive of forecasters to manipulate the decision in order to gain higher prediction rewards.
\begin{definition} \label{def:alt-indep}
  An aggregation mechanism is \emph{alternative-independent} if the expected reward $u^{i,1}$ of any forecaster $i$ does not depend on the distribution $\vect{p}$ from which the final alternative is chosen, as long as $\vect{p}$ has full support.
\end{definition}
We will see that an importance-weighting technique due to \citet{Chen11} allows one to achieve alternative-independence, so that prediction rewards are independent of the decisionmaking mechanism (at least in expectation.)
Under alternative independence, a forecaster does not gain increased information rewards by changing the decisionmaking mechanism's output $\vect{p}$ because her utility is now only a function of her prediction, i.e. $u^{i,1}(\vect{\hat{b}}^i)$.
It follows that an agent's strategy $\sigma^{i,2}$ in any QTM subgame is to maximize utility in that complete-information subgame exactly as in Section \ref{sec:external-welfare}, when external welfare was fixed and known.

Next is the incentive of decisionmakers to manipulate their predictions $\vect{\hat{b}}$ in order to boost the chances of their preferred alternative(s).
In general, agents can gain from manipulating the estimated external welfare even at a cost.
However, we will be able to show that the QTM's welfare guarantees are robust to relatively small manipulations.
Therefore, we only need that large manipulations are prohibitively costly.
We formalize this costliness as follows.
Say that a player is \emph{$x$-best-responding} in a game if they can improve their net expected payoff $u^i$ by at most $x$ by switching to another strategy.
\begin{definition} \label{def:deviation-bound}
In an aggregation mechanism, let $\vect{B^*}$ be the output of the mechanism if all forecasters are truthful and let $\vect{\hat{B}}$ be the output under strategic behavior.
Given some $x > 0$ and $\alpha > 0$, we say the aggregation mechanism is $(\alpha, x)$-robust if, for all strategy profiles $\vec{\sigma}$ where all participants are $x$-best-responding,
  \[ \max_k |\hat{B}_k - B_k^*| \leq \alpha x  . \]
\end{definition}

\paragraph{Solution concept.}
Because SQUAP is a dynamic game, we consider a refinement of Nash equilibria: a strategy profile $\vec{\sigma}$ is in \emph{sub-game perfect pure equilibrium} if it is in pure Bayes Nash equilibrium and, in addition, the strategies $\sigmatwo$ played at each subgame consisting of the Synthetic Players QTM are in pure Nash equilibrium.
In a setting defined by valuations $\vect{v}$ and information structure $\D$, we let $\text{pPNE}(\vect{v},\D)$ denote the set of distributions $\vect{p} \in \Delta_m$ induced by the mechanism in any sub-game perfect pure equilibrium.
Then we define $\text{pPoA}(\vect{v},\D)$ as the minimum ratio of the welfare in sub-game perfect pure equilibrium to $W_1 = V_1 + B_1^*$, i.e. 
\begin{equation}
	\text{pPoA}{(\vect{v}, \D)} = \frac{\min_{\vect{p} \in \text{pPNE}(\vect{v}, \D)} \sum_k p_k W_k} {W_1}. \label{eq:poa-aggr}
\end{equation}

\paragraph{PoA results.}
To prove our main Price of Anarchy results, we first use the strategic properties of alternative-independence and deviation robustness to show that, in equilibrium, the estimates $\vect{\hat{B}}$ are reasonably accurate.
\begin{lemma} \label{lem:aggr-bhat}
	In any sub-game perfect pure equilibrium of SQUAP, if the aggregation mechanism is alternative-independent and $(\alpha, x)$-robust, then its output satisfies for all $k$ the bound $|\hat{B}_k - B_k^*| \leq \alpha \max_{i,k} v_k^i$.
\end{lemma}
\begin{proof}

	Fix any events of the aggregation mechanism and consider the subgame of the QTM.
	Recall that an agent's expected utility in the combined mechanism is $u^i(\vect{\hat{b}}, \vect{a}, \vect{p}) = u^{i,1}(\vect{\hat{b}}^i) + u^{i,2}(\vect{a}, \vect{p})$.
	By alternative-independence, $u^{i,1}(\vect{\hat{b}}^i) = C$ for some constant $C$ independent of $\vect{p}$. 
	Meanwhile, by design of the QTM without redistribution, $u^{i,2}(\vect{a}, \vect{p}) \leq \max_{i,k} v_k^i$ for any $\vect{p}$.
	It immediately follows that in any QTM subgame equilibrium, all agents are ($\max_{i, k} v_k^i$)-best-responding. 

	Now we aim to show that in any combined equilibrium, in the aggregation stage all agents are ($\max_{i, k} v_k^i$)-best-responding. 
	We proceed by contradiction: assume there is some agent who is not.
	By definition, the agent can deviate to a strategy in the aggregation mechanism which leads to a utility increase strictly larger than $\max_{i, k} v_k^i$.
	Moreover, if the agent deviates to voting $\hat{\vect{a}}^i = \vect{0}$ in the QTM, she loses at most $\max_{i, k} v_k^i$ in utility because (1) $u^{i,1}$ cannot change by alternative independence and (2) $u^{i,2}(\vect{a}, \vect{p}) - u^{i,2}(\vect{a}^{-i}, \hat{\vect{a}}^i, \vect{\hat{p}}) \leq \max_{i, k} v_k^i$ for any two vectors $\vect{p}, \vect{\hat{p}}$.
	It follows that this deviation strategy leads to a strictly positive change in utility; thus, the current strategy is not an equilibrium. 
	By contradiction, then, the statement holds. 
	Since all agents are $(\max_{i,k} v_k^i)$-best responding and the aggregation mechanism is $(\alpha, x)$-robust, the lemma follows.
\end{proof}

We are now able to bound on the Price of Anarchy in the $m = 2$ case by modifying our QTM Price of Anarchy analysis.
Note that the incentive for manipulation, captured by the size of $\alpha$ in $(\alpha,x)$-robustness, can grow larger slowly as the game grows large and the Price of Anarchy will still tend to one.
\begin{theorem} \label{thm:general-aggr}
	Define the \emph{spread} of the game to be $T \coloneqq \frac{W_1}{max_{i, k} v_k^i}$.
	Suppose we run SQUAP with an aggregation mechanism that is alternative-independent and $(\alpha, x)$-robust.
    Then in the two-alternative setting, there exists a choice of $c$ such that
		\[ \text{pPoA}(\vect{v},\D) \geq 1 - \frac{2\alpha}{T} - \left(\frac{4}{T}\right)^{2/5} . \] 
\end{theorem}

The general strategy of proving Theorem \ref{thm:general-aggr} mirrors our QTM analysis, where we consider the case of a large welfare gap $W_1 - W_2$.
We first prove:
\begin{lemma} \label{lem:aggr-gap-poa}
	If the aggregation mechanism is alternative-independent and $(\alpha, x)$-robust, and $W_1 - W_2 \geq 2 \alpha \max_{i,k} v_k^i$, then in any submechanism-perfect pure equilibrium of the two alternative setting,
		\[ p_1 \geq 1 - \left(\frac{8c}{W_1 - W_2 - 2 \alpha \max_{i,k} v_k^i}\right)^{2/3} . \]
\end{lemma}
\begin{proof}
	Define $\hat{W}_k = V_k + \hat{B}_k$, and recall $W_k = V_k + B^*_k$.
	By Lemma \ref{lem:aggr-bhat}, we have
	\begin{align*}
		\hat{W}_1 - \hat{W}_2
			&= W_1 - W_2 + \hat{B}_1 - B^*_1 + \hat{B}_2 - B^*_2  \\
			&\geq W_1 - W_2 - 2 \alpha \max_{i,k} v_k^i  \\
			&\geq 0 .
	\end{align*}
	In the QTM stage, the synthetic players' aggregate values are $\vect{\hat{B}}$.
	As mentioned before, under alternative independence forecasters will not manipulate their votes to increase utility in the aggregation stage.
	Thus, agents in equilibrium will vote strategically in the QTM stage exactly as in Section \ref{sec:external-welfare}.
	It follows by Corollary \ref{cor:synthetic-fixed-ppoa-bb} that
	\begin{align*}
		p_1 &\geq 1 - \left(\frac{8c}{\hat{W}_1 - \hat{W}_2}\right)^{2/3}  \\
			&\geq 1 - \left(\frac{8c}{W_1 - W_2 - 2\alpha \max_{i,k} v_k^i}\right)^{2/3} .
	\end{align*}
\end{proof}

\begin{proof}[Proof of Theorem \ref{thm:general-aggr}]
	As with the QTM analysis, we divide into cases.
	Let $c = \tfrac{1}{2}\max_{i,k} v_k^i$.
	Let $Y$ be the magic number $(8c)^{2/5}W_1^{3/5}$ that balances the following two cases.

	Case $W_1 - W_2 \leq 2\alpha \max_{i,k} v_k^i + Y$:
	In this case, the Price of Anarchy is at least
	\begin{align*}
		\frac{W_2}{W_1} &= \frac{W_1 - (W_1 - W_2)}{W_1}  \\
		&\geq 1 - \frac{2\alpha \max_{i,k} v_k^i + Y}{W_1}  \\
		&= 1 - \frac{2\alpha}{T} - \left(\frac{4}{T}\right)^{2/5} .
	\end{align*}

	Case $W_1 - W_2 \geq 2\alpha \max_{i,k} v_k^i + Y$:
	In this case, by Lemma \ref{lem:aggr-gap-poa}, the Price of Anarchy is at least
	\begin{align*}
		&1 - \left(\frac{8c}{W_1 - W_2 - 2\alpha \max_{i,k}}\right)^{2/3}  \\
		&\geq 1 - \left(\frac{8c}{Y}\right)^{2/3}  \\
		&= 1 - \left(\frac{4}{T}\right)^{2/5} .
	\end{align*}
	In both cases, $\text{pPoA} \geq 1 - \frac{2\alpha}{T} - \left(\frac{4}{T}\right)^{2/5}$.
\end{proof}

Equipped with Lemma \ref{lem:aggr-bhat}, we can also give our Price of Anarchy result for $m$ alternatives.
This result is phrased as a reduction from any Price of Anarchy bound for the QTM under the same scenario but without predictions.
The proof is similar to the above and appears in Appendix \ref{app:predictions}.
\begin{theorem} \label{thm:general-aggr-general-m}
	Define the \emph{spread} of the game to be $T \coloneqq \frac{W_1}{max_{i, k} v_k^i}$.
	Suppose we run SQUAP with an aggregation mechanism that is alternative-independent and $(\alpha, x)$-robust, and suppose the QTM with $m$ alternatives has a Price of Anarchy guarantee $\text{pPoA}{(\vect{v})} \geq C(\vect{v})$.
	Then the Price of Anarchy of SQUAP satisfies
		\[ \text{pPoA}(\vect{v},\D) \geq \frac{1}{1 + \alpha} \left( C(\vect{v}) - \frac{\alpha}{T} \right). \] 
\end{theorem}

We now have a Price of Anarchy bound of Theorem \ref{thm:general-aggr} for any information-aggregation mechanism that is alternative-independent and $(\alpha, x)$-robust.
Next, we show that two such information aggregation mechanisms do actually exist.

\subsection{Wagering Mechanisms}
We consider using wagering mechanisms as the information-aggregation mechanism in stage 1 of SQUAP.
Wagering mechanisms are a common tool for eliciting and aggregating predictions~\citep{lambert2008self,chen2015axiomatic}.
For simplicity, we will suppose that agents predict $\vect{B}$ directly and the mechanism is later able to exactly observe $B_k$, the true external welfare of the alternative $k$ that was selected.
When observing $B_k$ is difficult, one could consider wagers that are more complex and predict other relevant variables.
In this case, the mechanism could compute an estimate of $\vect{B}$ from the mechanism's forecasts and use the estimate in SQUAP.

We consider a natural extension that does not appear to have been studied: \emph{decision wagering mechanisms}.
Instead of agents wagering on a single future event, we will have them wager on the future event conditional on alternative $k$ being selected, for each $k=1,\dots,m$.
Once $k$ is selected and the event is observed, we cancel all other wagers and assign payments based on the predictions for $k$.

Specifically, we adapt the Brier betting mechanism \citep{Lambert2015} to our setting, defining the \emph{importance-weighted quadratic decision wagering mechanism}.
The Brier betting mechanism makes use of the quadratic score $s$ with parameter $\beta > 0$:
\begin{equation} \label{eq:quadratic-score}
	s(\hat{b},b) = -\frac{1}{\beta}(\hat{b} - b)^2 .
\end{equation}
$s$ is one example of a \emph{proper scoring rule} for the mean~\citep{savage1971elicitation,gneiting2007strictly}: the expected value of $s(\hat{b},b)$, over draws $b$ from some distribution $\D$, is maximized by choosing $\hat{b} = \E[b]$.
Following the prediction market literature, we refer to $\beta$ as the liquidity parameter. 
$\beta$ determines how sensitive the payoffs are to small changes in a prediction. 
If $\beta$ is small, then the total payoffs available are very large, but a significant payment may be required as well.

Formally, for each alternative $k$, forecaster $i$ submits an estimate $\hat{b}^i_k$ for the expected external welfare conditioned on that alternative. 
Let $\hat{\vect{b}}$ consist of all the players' predictions for all the alternatives.
Once an alternative $k \sim \vect{p}$ is chosen and the outcome $b_k$ is observed, each player $i$ receives the payoff
  \[ \pi_i(\vect{\hat{b}}, \vect{p}, k,  b_k) = \frac{1}{p_k} \left[ s(\hat{b}^i_k, b_k) - \frac{1}{n}\sum_{j = 1}^n s(\hat{b}^j_k, b_k) \right] .\]
In other words, forecaster $i$'s payoff consists of her score minus the average score of all participants.
Since $\sum_i \pi_i(\vect{\hat{b}}, \vect{p}, k, b_k) = 0$, we obtain:
\begin{observation}
  The importance-weighted quadratic decision wagering mechanism achieves budget-balance.
\end{observation}

\paragraph{Wagers and agent loss.}
Above, we have considered the special case of a wagering mechanism where all wagers are equal to one.
In general, the wagering mechanism weights each participant's score $s(\hat{b}^i_k,b_k)$ by $\frac{y_i^k}{Y_k}$, where $y^i_k$ is some nonnegative wager on alternative $k$ and $Y_k = \sum_i y^i_k$.
The wagering mechanism is generally applied in cases where scores $s()$ are bounded in $[0,1]$, so that an agent can be guaranteed to win or lose no more than her wager.
In our case, such guarantees can be achieved by adjusting $\beta$, the liquidity parameter of $s$, if an upper bound on $B_k$ is known.
Wagering mechanisms can also be subsidized, e.g. by additionally giving each agent a fraction of their score, or by making all payouts nonnegative with a shift.
We could use a similar analysis to the prediction market setting to use the QTM revenue to subsidize the wagering mechanism (see Corollary~\ref{cor:prediction-market-bb}.)

\paragraph{Strategic behavior and aggregation.}
For analysis of this mechanism, unlike with prediction markets in Section \ref{subsec:prediction-markets}, we will need to adopt the standard wagering mechanism model of \emph{immutable beliefs}, where each agent $i$ has fixed beliefs $\vect{b}^i$ about the expectation of $\vect{B}$, regardless of others' beliefs.
(In a Bayesian model, we essentially have a ``no-trade'' situation in which the zero-sum payoffs discourage participation.)
The wagering mechanism in isolation is strategyproof: agents maximize expected net payoff by reporting their true believed estimates.
This follows immediately from the properness of the scoring rule $s$ and the fact that in $\pi_i$, the agent can only affect their own score.
For analysis, we suppose that the wagering mechanism outputs $\hat{B}_1,\dots,\hat{B}_m$ where $\hat{B}_k = \frac{1}{N} \sum_{i=1}^N \hat{b}_k^i$.
We assume that, if all agents are truthful, then the output $\vect{B^*}$ is the true expectation of the external welfare impacts.
Our analysis could be adapted to other aggregation methods than taking the average.

Now consider agent $i$. 
The only terms in her expected utility that $i$ controls are 
\begin{equation*}
	\frac{1}{p_k} \left( s(\hat{b}_k^i,b_k) - \frac{1}{n} s(\hat{b}_k^i,b_k) \right) = \left( \frac{n-1}{n} \right) \frac{1}{p_k} s(\hat{b}_k^i,b_k)
\end{equation*}
for each $k$.
If the outcome $b_k$ is drawn from a distribution $\D_k$, then the expected score over both $k \sim \vect{p}$ and $b_k \sim \D_k$ is
\begin{align*}
	\sum_k p_k \E_{\D_k} \left[ \left( \frac{n-1}{n} \right) \frac{1}{p_k} s(\hat{b}_k^i,b_k) \right] &= -\frac{n-1}{\beta n} \sum_k \left( (\hat{b}_k^i - b_k^i)^2 + \text{Var}(D_k) \right).
\end{align*}
The variance term is independent of any agents' actions, and moreover cancels out in an agent's payoff when we take the difference over different strategies. 
Therefore, it is without loss of generality to define the \emph{expected score} of a prediction $\vect{\hat{b}}$ to be
\begin{align}
	S(\vect{\hat{b}}^i; \vect{b}^i) &= -\frac{n-1}{\beta n} \sum_k (\hat{b}_k^i - b_k^i)^2 .  \label{eq:expected-score-wager}
\end{align}
Because (\ref{eq:expected-score-wager}) does not depend on the selected outcome $k$, we obtain:
\begin{observation} \label{obs:wagering-alt-indep}
	The importance-weighted quadratic decision wagering mechanism with the importance-weighted quadratic score is alternative-independent.
\end{observation}

\paragraph{Price of Anarchy.}
We now have the tools to analyze SQUAP with the importance-weighted quadratic decision wagering mechanism.
\begin{lemma} \label{lem:wagering-dev}
	The importance-weighted quadratic decision wagering mechanism with parameter $\beta = \frac{1}{2} \epsilon x$ is $(\epsilon^{1/2}, x)$-robust. 
\end{lemma}

\begin{proof}
	By Equation \ref{eq:expected-score-wager}, the loss from misreporting $\vect{\hat{b}}^i$ is
	\begin{align*}
		S(\vect{\hat{b}}^i;\vect{b}^i) - S(\vect{b}^i;\vect{b}^i) &= -\frac{n-1}{\beta n} \sum_k (\hat{b}_k^i - b_k^i)^2 \nonumber \\
		&\leq -\frac{1}{2 \beta} \sum_k (\hat{b}_k^i - b_k^i)^2,
	\end{align*}
	where the last line follows since $n \geq 2$. 

	By the definition of $x$-best-responding in the wagering mechanism,
	$x > \frac{1}{2 \beta} \sum_k (\hat{b}_k^i - b_k^i)^2$.
	So for all $k$, $(\hat{b}_k - b_k^i)^2 \leq 2 \beta x = \epsilon x^2$, or $|\hat{b}_k^i - b_k^i| \leq \epsilon^{1/2} x$.
	Now, this holds for all participants, and the output of the wagering mechanism is $\hat{B}_k = \frac{1}{n} \sum_{i=1}^n \hat{b}^i_k$, so by definition the mechanism is $(\epsilon^{1/2}, x)$-robust. 
\end{proof}

\begin{corollary} \label{cor:wagering-poa}
	Let $T = \frac{W_1}{\max_{i,k} v_k^i}$ be the spread.
	In the two alternative setting, SQUAP with the importance-weighted quadratic decision mechanism with wagering parameter $\beta = \frac{1}{2} \epsilon \max_{i,k} v_k^i$ and QTM parameter $c = \tfrac{1}{2} \max_{i,k} v_k^i$ satisfies
		\[ \text{pPoA}(\vect{v}, \D) \geq 1 - \frac{2\epsilon^{1/2}}{T} - \left(\frac{4}{T}\right)^{2/5} . \]
\end{corollary}

We observe that, even with constant liquidity $\epsilon = \Theta(1)$, the Price of Anarchy tends to one with the spread.
When $m > 2$ and the QTM with $m$ alternatives has a Price of Anarchy guarantee $\geq C(\vect{v})$, we could also apply Theorem~\ref{thm:general-aggr-general-m} to recover a Price of Anarchy guarantee $\frac{1}{1 + \epsilon^{1/2}}\left( C(\vect{v}) - \frac{\epsilon^{1/2}}{T}\right)$.

\subsection{Prediction Markets} \label{subsec:prediction-markets}
Next we consider using prediction markets as the information-aggregation mechanism in stage $1$ of SQUAP. 
Prediction markets may be considered more complex than wagering mechanisms since they are not one-shot.
However, while the designer must combine the resulting predictions from a wagering mechanism herself (often by averaging), prediction markets automatically aggregate information. 

Prediction markets are financial markets designed specifically for aggregating the predictions of agents by rewarding accuracy once an outcome is observed.
We briefly define scoring-rule based prediction and decision markets below, referring the reader to references in Section \ref{sec:related-work} for more background.

\paragraph{The market scoring rule.}
We define a scoring-rule prediction market for a real-valued random variable, following e.g. \citet{abernethy2013efficient}.
The market defines an initial estimate $\hat{b}^0$. Agents one at a time arrive in sequence and provide estimates $\hat{b}^1,\dots,\hat{b}^N$.
Later, when the outcome $b$ of the random variable is observed, the provider of each estimate $\hat{b}^t$ receives a net payoff $s(\hat{b}^t,b) - s(\hat{b}^{t-1},b)$, where $s$ is some proper scoring rule as in the wagering mechanism setting.
In particular, we will again focus on the quadratic scoring rule as defined in Equation \ref{eq:quadratic-score}.

Observe that each participant's net payoff can be positive or negative, and the total payment of the mechanism telescopes to $s(\hat{b}^N,b) - s(\hat{b}^0,b)$.
Following the prediction market literature and our section on wagering mechanisms, $\beta$ represents the liquidity parameter.
We have defined the market so that agents arrive only once each, but our results are robust to agents participating multiple times or at strategically chosen times, subject to an assumption discussed next.

\paragraph{Solution concept.}
Since the prediction market framework is not a one-shot game, but instead an extensive-form Bayesian game, we must modify our solution concept from subgame perfect equilibrium.
We say the combined mechanism is in \emph{sub-mechanism perfect equilibrium} if it is in Nash equilibrium and, in addition, the strategies played at each subgame consisting of the Synthetic Players QTM are in pure-strategy Nash equilibrium.
However, we will need further assumptions on the strategies played in the game as well.

Strategic behavior in prediction markets is highly complex and not fully understood~\citep{chen2010gaming,ostrovsky2012information,gao2013jointly,chen2016informational}.
However, one robust result of all of the above works is that, in the long run, agents do not stay misinformed.
They may be \emph{under}-informed, but indeed all information is often aggregated in equilibrium (with an appropriate solution concept such as Weak Perfect Bayesian equilibrium).
In particular, it is known that in strategic equilibrium of prediction markets, information is always aggregated under a condition called \emph{separable securities}~\citep{ostrovsky2012information}.
In other words, when information aggregation fails in prediction markets, it is not for strategic reasons, but only because of the expressiveness of the prediction language.
Therefore, we believe that efficient markets is a mild assumption on strategic behavior of the agents in the combined mechanism, and it can be achieved by introducing additional tradeable securities.
\begin{definition} \label{def:efficient-markets}
	We say the \emph{efficient markets assumption} holds on the decision market if, in any equilibrium, in any subgame (including off the equilibrium path), the last participant in the decision market knows and believes $\vect{B^*}$, the true expected external welfare impacts of the alternatives conditioned on all available information.
\end{definition}
In other words, we assume that a manipulator is not able to cause information to be hidden or the other participants to be misled.
However, a manipulator can still change the final predictions by waiting until the final stage to trade.

The impact of the efficient markets assumption is that a manipulator is best off if they are the one predicting last.
She waits for the markets to converge to the prediction $\vect{\hat{B}} = \vect{B^*}$, then manipulates $\vect{\hat{B}}$ arbitrarily.
We will show that such a manipulation will not be large enough to damage the welfare guarantee of SQUAP.

\paragraph{Decision markets.}
In a decision market (e.g. \citet{Chen11}), we simultaneously run $m$ prediction markets, one for each of $m$ alternatives.
Here, we model this as each participant $t=1,\dots$ providing an estimate $\vect{\hat{b}}^t \in \R^m$.
Then, we make a decision $k$ based on the outcome of the markets, and later observe $b_k$.
We cancel all trades in the other $m-1$ markets, i.e. we only assign payoffs based on predictions $\hat{b}_k^t$ made in the $k$th prediction market.

As we saw with decision wagering mechanisms, decision markets can lead to incentives to misreport and manipulate the decision.
We thus use the same importance-weighted scoring rule proposed by \citet{Chen11}:
	\[ S(\vect{\hat{b}}^t, \vect{p}, k, b_k) = \frac{1}{p_k} s(\hat{b}^t_k, b_k) . \]
Then if the outcome $b_k$ is drawn from a distribution $D_k$, the expected score over both $k \sim \vect{p}$ and $b_k \sim D_k$ is
\begin{align*}
	\sum_k p_k \E_{\D_k} S(\vect{\hat{b}}^t, \vect{p}, k, b_k)
	&= \sum_k \E_{\D_k} (\hat{b}_k^t - b_k)^2 \\
    &= \frac{-1}{\beta} \sum_k \left( (\hat{b}_k^t - b^*_k)^2 + \text{Var}(\D_k) \right),
\end{align*}
where $b_k^* = \E_{D_k} [b_k]$.
For the same reasons as before, we can ignore the variance term since it is independent of agent actions and cancels out in the market scoring rule.
Therefore, under the importance-weighted scoring rule, it is without loss of generality to define the \emph{expected score} of a prediction $\vect{\hat{b}}$, given the vector of means $\vect{b^*} = b_1^*,\dots,b_m^*$, to be
\begin{align}
	S(\vect{\hat{b}}; \vect{b^*}) &= \frac{-1}{\beta} \sum_k (\hat{b}_k - b_k^*)^2 .  \label{eq:expected-score}
\end{align}
We obtain:
\begin{observation} \label{obs:market-alt-indep}
	The decision market with the importance-weighted quadratic score is alternative-independent.
\end{observation}

\subsubsection{Price of Anarchy}
We now consider the welfare guarantees of SQUAP with the decision market aggregation mechanism.
We leave the proof, which follows a similar pattern as the analogous result for wagering mechanisms, in Appendix~\ref{app:prediction-market}.
\begin{lemma} \label{lem:market-deviation}
	Under the efficient market assumption, for any $x > 0$, the importance-weighted decision market with the quadratic scoring rule and liquidity parameter $\beta = \epsilon x$ is $(\epsilon^{1/2}, x)$-robust. 
\end{lemma}

From Theorem \ref{thm:general-aggr}, we immediately get the following.
\begin{corollary} \label{cor:market-poa}
	Let $T = \frac{W_1}{\max_{i,k} v_k^i}$ be the spread.
	In the two alternative setting and under the efficient markets assumption, SQUAP with market parameter $\beta = \epsilon \max_{i,k} v_k^i$ and QTM parameter $c = \tfrac{1}{2} \max_{i,k} v_k^i$ satisfies
		\[ \text{pPoA}(\vect{v},\D) \geq 1 - \frac{2\epsilon^{1/2}}{T} - \left(\frac{4}{T}\right)^{2/5} . \]
\end{corollary}
As in the wagering mechanism setting, even with constant liquidity $\epsilon = \Theta(1)$, the Price of Anarchy tends to one with the spread.
Typically in a prediction market, we have $\epsilon \to 0$ with the size of the market~\citep{abernethy2014general}.
We would expect $\epsilon \to 0$ in large settings in practice, but it is not required for good Price of Anarchy.
This occurs because the QTM is robust to rougher estimates as the spread of the market grows.
While one might prefer a larger decision market in practice in hopes that it better aggregates information, this robustness may at least be reassuring.
We will also utilize it next to obtain budget balance in many cases.

\subsubsection{Budget balance}
Prediction markets are typically assumed to be subsidized.
Although the subsidy could come from charging transaction fees in the market, a natural question is whether the transfers of the QTM portion of the mechanism can be used to fund the aggregation mechanism.
In general, this is not always achievable, because as shown by \citet{Eguia19,Eguia23}, the revenue of the QTM stage can shrink to zero even as the spread of the market grows.
However, Proposition \ref{prop:revenue} allows us to characterize settings with growing revenue, therefore leading to a budget balanced mechanism that self-funds its decision market.
\begin{corollary} \label{cor:prediction-market-bb}
  Let the disagreement of a value profile $\vect{v}$ be $D := \frac{\sum_i (v_1^i - v_2^i)^2}{(W_1 - W_2)^2}$.
  Fixing the information structure $\D$, in a sequence of SQUAP mechanisms with growing spread $T \to \infty$ and bounded-below disagreement $D \geq \Omega(1)$, the mechanism does not lose money in expectation and $\text{pPoA} \to 1$.
\end{corollary}

\section{Discussion} \label{sec:discussion}

\paragraph{Summary.}
This paper's first contribution was a Price of Anarchy analysis of the Quadratic Transfers Mechanism, primarily in the two-alternative setting.
Prior work analyzed welfare in asymptotic settings with randomly drawn agents.
But results for such a setting, reminiscent of a ``strategic jury theorem,'' left open a question of robust guarantees with arbitrary sets of agents, as we gave in Section \ref{sec:qtm}.

The second contribution was to extend the QTM to a case with external welfare impacts of the decisions.
Many public projects settings involve impacts on non-decisionmakers, such as the climate impacts of a particular policy.
Section \ref{sec:external-welfare} proposed a mechanism for this problem, the Synthetic Players QTM, for which we provided strong Price of Anarchy guarantees. 
In particular, the mechanism approximately maximizes the welfare of decisionmakers plus externalities of the decision.
We also suggest two variants of the mechanism which could be more practically implemented by the designer. 

The third and main contribution was to propose decisionmaking mechanisms that combine \emph{preference aggregation} with \emph{information aggregation}.
We model the information aggregation stage as forecasting of the external welfare impacts of the decision.
By using either prediction markets or wagering mechanisms, we again showed that the the Synthetic Players QTM with Predictions (SQUAP) has strong Price of Anarchy guarantees.
The importance-weighting technique gave forecasters the same expected payoff for any decision of the mechanism.
The proper scoring rules enforced penalties for inaccurate predictions, which can only change the final decision with a small probability while causing a prohibitive penalty.

\paragraph{Future work.}
There are a number of technical open problems.
First is improving the QTM Price of Anarchy guarantee for more than two alternatives, or giving counterexamples for parameter regimes where it is impossible.
Second is proving a Price of Anarchy guarantee for the practical variants of the Synthetic Players QTM, or finding counterexamples.
Even in the latter case, we expect a ``Price of Stability'' result to be possible.
Third is to extend SQUAP to other information-aggregation mechanisms, perhaps with a more practical method than importance weighting, and prove guarantees for the variant.

Conceptually, there are many other possible approaches to decisionmaking with preferences and predictions.
In particular, our approach used monetary mechanisms and the quantitative criterion of social welfare maximization, which required a strong quasilinear assumption on participants.
Even within that sphere, we have only proposed one possible approach.
But outside of it, many other directions are possible.

\subsection*{Acknowledgements} 

Supported by the Ethereum Foundation grant FY22-0716.

\bibliographystyle{ACM-Reference-Format}
\bibliography{citations}

\break

\appendix
\section{Quadratic Transfer Mechanism} \label{app:qtm}
This section gives proofs from Section \ref{sec:qtm} along with some additional results and discussion.

\subsection{Equilibrium Characterization} \label{app:equilibrium}
In this section, we characterize conditions under which a pure-strategy equilibrium exists in the QTM.
Our results build on some techniques of ~\citet{Eguia19}, but also introduce new techniques and new structural features of the mechanism, such as the fact that votes across alternatives always sum to zero in equilibrium.

To show existence, we take the following steps.
First (Lemma \ref{lem:strictly-dominated-strategy}), we show that for each player, it is strictly dominated to play any $\vect{a}^i$ lying outside a certain compact set.
It follows that all equilibria are ``as if'' the game were restricted to a compact strategy space (Lemma \ref{lem:compact-to-unbounded-equilibrium}).
From there, we can show that if the parameter $c$ is large enough relative to the maximum agent preference, then utilities are concave, and apply an existence theorem (Proposition \ref{prop:equilibrium-existence}).
A similar argument appears in \citet{Eguia19}, but in our case we give a constructive bound on $c$ that appears quite tight.
We also note a connection in the analysis to exponential family distributions.
Finally, we derive the first-order conditions for pure-strategy Nash equilibrium, and show as a byproduct that each agent's votes sum to zero.

To begin, recall that a strategy is \emph{strictly dominated} if there exists another strategy that always obtains strictly higher utility.
\begin{lemma} \label{lem:strictly-dominated-strategy}
	For all agents $i$, if the pure strategy $\vect{a}^i$ has $|a_k^i| >  \sqrt{ \frac{max_{l} v_l^i}{c} }$ for any $k$, then $\vect{a}^i$ is a strictly dominated strategy.
\end{lemma}
\begin{proof}
	Fix all other agents' strategies as $\vect{a}^{-i}$. Denote the expected utility $u^i$ of pure strategy $\vect{a}^i$ given fixed strategies $\vect{a}^{-i}$ as $u^i(\vect{a}^i; \vect{a}^{-i})$, and the softmax probability $p_k$ over strategy $a^i$ and fixed strategies $\vect{a}^{-i}$ as $p_k(\vect{a}^i; \vect{a}^{-i})$.
	
	Consider the strategy $\hat{a}_k^i = 0 \; \forall k \in \alternativeset$. Then
	\begin{align}
		u^i(\vect{\hat{a}}^i; \vect{a}^{-i}) &= \sum_{k=1}^m p_k(\vect{\hat{a}}^i; \vect{a}^{-i}) v_k^i - c \sum_{k=1}^m 0^2 + \frac{1}{n-1} \sum\limits_{\substack{j=1 \\ j \neq i}}^n \sum_{k = 1}^m (a_k^j)^2 \nonumber \\
		&\geq \frac{1}{n-1} \sum\limits_{\substack{j=1 \\ j \neq i}}^n \sum_{k = 1}^m (a_k^j)^2. \label{eq:dominating-strategy}
	\end{align}
	Now, take a different pure strategy $\vect{\tilde{a}}^i$ for voter $i$ where $|\tilde{a}_k^i| > \sqrt{ \frac{max_{l} v_l^i}{c} }$ for some $k \in \alternativeset$. 
	
	Since $(\tilde{a}_k^i)^2$ is nonnegative, we must have $\sum_{k=1}^m (\tilde{a}_k^i)^2 > \frac{1}{c} (\max_l v_l^i)$. Thus
	\begin{align}
		 u^i(\vect{\tilde{a}}^i; \vect{a}^{-i}) &= \sum_{k=1}^m p_k(\vect{\tilde{a}}^i; \vect{a}^{-i}) v_k^i - c \sum_{k=1}^m (\vect{\tilde{a}}_k^i)^2 + \frac{1}{n-1} \sum\limits_{\substack{j=1 \\ j \neq i}}^n \sum_{k = 1}^m (a_k^j)^2 \nonumber \\
		&< \max_l v_l^i - \max_l v_l^i + \frac{1}{n-1} \sum\limits_{\substack{j=1 \\ j \neq i}}^n \sum_{k = 1}^m (a_k^j)^2 \nonumber \\
		&= \frac{1}{n-1} \sum\limits_{\substack{j=1 \\ j \neq i}}^n \sum_{k = 1}^m (a_k^j)^2. \label{eq:dominated-strategy}
	\end{align}
	Combining Inequalities \ref{eq:dominating-strategy} and \ref{eq:dominated-strategy}, for any strategy profile $\vect{a}^{-i}$ of voters other than $i$, $u^i(\vect{\tilde{a}}^i; \vect{a}^{-i}) < \frac{1}{n-1} \sum_{j \neq i} \sum_{k = 1}^m (a_k^j)^2 < u^i(\vect{\hat{a}}^i; \vect{a}^{-i})$.
	By definition, then, it is a strictly dominated strategy for any agent $i$ to play $\vect{\tilde{a}}^i$.
	
\end{proof}

\begin{lemma} \label{lem:compact-to-unbounded-equilibrium}
	Fix a set of agents with types $\vect{v}$.
	Let QTM' be the QTM where each agent is restricted to the compact strategy space $\Bigl[ -\sqrt{ \frac{\max_{i,k} v_k^i}{c} }, \sqrt{ \frac{\max_{i,k} v_k^i}{c} } \Bigr]^m$.
	A pure-strategy Nash equilibrium exists for QTM' if and only if that same pure-strategy equilibrium exists for QTM.
\end{lemma}

\begin{proof}
	Let $\mathcal{A} := \Bigl[ -\sqrt{ \frac{\max_{i,k} v_k^i}{c} }, \sqrt{ \frac{\max_{i,k} v_k^i}{c} } \Bigr]^m$.
	Suppose $\vect{a}_* = (\vect{a}^1_*, \vect{a}^2_*, ... \vect{a}^n_*)$ is a pure-strategy equilibrium for QTM'.
	Then for each $i$, $\vect{a}_*^i$ is a best response among $\mathcal{A}$ to $\vect{a}_*^{-i}$. 
	By Lemma \ref{lem:strictly-dominated-strategy}, no strategy $\vect{a}^i \not\in \mathcal{A}$ can be a best response to $\vect{a}_*^{-i}$, because $\vect{a}^i$ is strictly dominated by some strategy in $\mathcal{A}$.
	So $\vect{a}_*$ remains an equilibrium for QTM.

	Conversely, let $\vect{a}_*$ be an equilibrium for QTM.
	We must have $\vect{a}_*^i \in \mathcal{A}$ for all $i$, as strictly dominated strategies are never best responses.
	So $\vect{a}_*$ is an equilibrium for QTM', as the strategies all lie in $\mathcal{A}$ and are all best responses among all strategies in $\mathcal{A}$.
\end{proof}

\begin{lemma} \label{lem:concave-utils}
	In the QTM on agent value profile $\vect{v}$, if the mechanism chooses $c \geq \tfrac{1}{2} \max_{i,k} v_k^i$, then every agent's utility is strictly concave as a function of their vote $\vec{a}^i$, for any fixed strategy of their opponents.
\end{lemma}

\begin{proof}[Lemma \ref{lem:concave-utils}]
	Note that $\hat{u}^i = u^i / c$ has the same concavity as $u^i$. 
	Consider $H^i$, the Hessian matrix of $\hat{u}^i$.
	We begin by computing expressions for diagonal entries $H_{kk}^i$ and non-diagonal entries $H_{kl}$, where $k \neq l$.
	\begin{align}
		\frac{\partial \hat{u}^i}{\partial a_k^i}                           & = \sum_{h=1}^{m} \frac{\partial}{\partial a_k^i}p_h \frac{v_h^i}{c} - 2a_k^i, \text{so}                \nonumber    \\
		H_{kk}^i = \frac{\partial^2 \hat{u}^i}{\partial a_k^i \partial a_k^i} & = \sum_{h=1}^{m} \frac{\partial^2}{(\partial a_k^i)^2}p_h \frac{v_h^i}{c} - 2 \label{eq:Uideriv1}        \\
		H_{kl}^i = \frac{\partial^2 \hat{u}^i}{\partial a_k^i \partial a_l^i} & = \sum_{h=1}^{m} \frac{\partial^2}{\partial a_k^i \partial a_l^i}p_h \frac{v_h^i}{c}. \label{eq:Uideriv2} 
	\end{align}
	The first derivatives for each probability $p_k$ are listed below:
	\begin{align*}
		\frac{\partial p_k}{\partial a_k^i} & = \frac{e^{A_k} \sum_{h=1}^{m} e^{A_h} - e^{2A_k}}{\Bigl(\sum_{h=1}^{m} e^{A_h}\Bigr)^2} = p_k - p_k^2 \\
		\frac{\partial p_l}{\partial a_k^i} & = \frac{- e^{A_l}e^{A_k}}{\Bigl(\sum_{h=1}^{m} e^{A_h}\Bigr)^2} = -p_l p_k  .                           
	\end{align*}
	We extend to the second derivatives as well; given $k \neq l \neq h$,
	\begin{align*}
		\frac{\partial^2 p_k}{\partial a_k^i \partial a_k^i} & = p_k (1 - 3p_k + 2p_k^2) = p_k (2p_k - 1) (p_k - 1) \\
		\frac{\partial^2 p_k}{\partial a_k^i \partial a_l^i} & = p_kp_l(2p_k - 1)                                   \\
		\frac{\partial^2 p_k}{\partial a_l^i \partial a_h^i} & = 2p_lp_hp_k                                         \\
		\frac{\partial^2 p_k}{\partial a_l^i \partial a_l^i} & = p_kp_l(2p_l - 1).
	\end{align*}
	Denote $\expvi = \sum_{k=1}^m p_k v_k^i$, the expected value of the mechanism for agent $i$. Plugging the above values into Equations (\ref{eq:Uideriv1}) and (\ref{eq:Uideriv2}), we have
	\begin{align*}
		H_{kk}^i & = \sum_{h=1}^{m} \frac{\partial^2}{(\partial a_k^i)^2}p_h \frac{v_h^i}{c} - 2                       \\
			& = \frac{p_k}{c} \Bigl[ (2p_k - 1) (p_k - 1) v_k^i + (2p_k - 1) \sum_{h \neq k} p_h v_h^i \Bigr] - 2 \\
			& = \frac{p_k}{c} (2p_k - 1) (\expvi - v_k^i) - 2,                                          
	\end{align*}
	\begin{align*}
		H_{kl}^i & = \sum_{h=1}^{m} \frac{\partial^2}{\partial a_k^i \partial a_l^i}p_h \frac{v_h^i}{c}                                    \\
			& = \frac{1}{c} \Bigl[ p_k p_l (2p_k - 1) v_k^i + p_k p_l (2p_l - 1)v_l^i + \sum_{h \neq k, l} 2 p_l p_h p_k v_h^i \Bigr] \\
			& = \frac{p_k p_l}{c} (2\expvi - v_k^i - v_l^i).                                                                    
	\end{align*}
	We consider the matrix $B^i$ such that $H^i = -2I + B^i$.
	Then the inequality $v^T H^i v < 0$ corresponds to $v^T B^i v < 2$: so to prove $H^i$ is negative definite, we just need to bound the spectral radius of $B^i$ by 2. 

	As an interesting side note, the sum over column $k$ in $B^i$ is
	\begin{align*}
		\biggl( \sum_{l \neq k} \frac{\partial^2 \hat{u}^i}{\partial a_k^i \partial a_l^i} \biggr) +  \frac{\partial^2 \hat{u}^i}{\partial a_k^i \partial a_k^i} + 2 & = \biggl( \sum_{l \neq k} \frac{p_k p_l}{c} (2\expvi - v_k^i - v_l^i) \biggr) + \frac{p_k}{c} (2p_k - 1) ( \expvi - v_k^i) = 0. 
	\end{align*}

	Now, the spectral radius of a matrix is bounded by the maximum sum of absolute value of entries in any row (in our case, $\sum_{l=1}^m |B^i_{kl}|$).
	Note that $|B_{kk}^i| = |H_{kk}^i + 2| \leq \tfrac{p_k}{c}(2p_k - 1) \max_k{ v^i_k}$ and  $|B_{kl}| = |B_{lk}^i| = |H_{kl}^i| \leq (2 p_k p_l \max_k{ v^i_k}) / c$.
	So
	\begin{align*}
		\sum_{l=1}^m |B_{kl}^i| & \leq \frac{2 p_k  \max_k{ v^i_k}}{c} \sum_{l \neq k} p_l + p_k (2p_k - 1) \frac{ \max_k{ v^i_k}}{c} \\
						& = \frac{ \max_h{ v^i_h}}{c}( 2p_k (1 - p_k) + p_k (2p_k - 1))                               \\
						& = \frac{p_k  \max_h{ v^i_h}}{c} < \frac{\max_h{ v^i_h}}{c}.            
	\end{align*}
	Thus, for $c \geq \frac{\max_h{v^i_h}}{2}$, the maximum sum of the absolute value of entries in any row of $B^i$ is strictly less than 2; the spectral radius of $B^i$ is strictly less than 2; and $H^i$ is negative definite for this value of $c$, meaning that $u^i$ is strictly concave.
\end{proof}

Now, we have the tools to prove Proposition \ref{prop:equilibrium-existence}, that for $c \geq \tfrac{1}{2} \max_{i,k} v_k^i$ a pure-strategy Nash equilibrium exists:

\begin{proposition} \label{prop:equilibrium-existence}
	In the QTM on agent value profile $\vect{v}$, if the mechanism chooses $c \geq \tfrac{1}{2} \max_{i,k} v_k^i$, then a pure-strategy Nash equilibrium exists.
\end{proposition}
\begin{proof}[Proposition \ref{prop:equilibrium-existence}]
	By Lemma \ref{lem:concave-utils}, if the mechanism chooses $c \geq \tfrac{1}{2} \max_{i,k} v_k^i$ then agent utilities are strictly concave as a function of their votes. 
	Since the strategy spaces of each agent are non-empty, compact, convex subsets of a Euclidean space; utilities are continuous; and utilities are strictly concave, it follows that a pure strategy Nash equilibrium exists \citep{Glicksberg52,Debreu52}.
\end{proof}

\begin{remark}
	The above analysis can be reinterpreted in terms of exponential family distributions.
	Specifically, the mechanism generates probabilities representing which one of $m$ events $X_1, X_2, \ldots, X_m$ will occur, and thus corresponds to a multinomial distribution with parameter $\vect{p}$.
	An exponential family distribution's $n$th cumulant can be calculated directly as the $n$th derivative of the cumulant function $A$.
	One can verify that the derivative of $A$ with respect to the $k$th parameter corresponds to $p_k$, the mean of the random variable $X_k$. 
	It follows that the first partial derivatives of $p$ correspond to the variance and covariance, and the second partial derivatives correspond to the third cumulant. 
\end{remark}

There may be many scenarios where a pure strategy equilibrium exists even for much lower values of $c$.
For example, we conjecture this to be the case when agent values are ``well-aligned'' with each other, without much disagreement.

We now state the important first-order conditions that characterize equilibria in terms of the mechanism's expected individual and aggregate welfare, $\expvi = \sum_k p_k v_k^i$ and $\expv = \sum_k p_k V_k$ respectively.
This result also appears throughout \citet{Eguia19}, although we add some observations that are useful to us later in constructing equilibria as part of our mechanisms.

\begin{lemma} \label{lem:foc}
  In any pure-strategy Nash equilibrium of the QTM with values $\vect{v}$, the votes $\vect{a}$ satisfy
  \begin{align}
    a_k^i &= \frac{p_k}{2c} \left(v_k^i - \sum_{\ell} p_{\ell} v_{\ell}^i \right)  & \text{for all $i$,}  \label{eq:foc} \\
    A_k   &= \frac{p_k}{2c} \left(V_k - \sum_{\ell} p_{\ell} V_{\ell} \right) .    \label{eq:aggregate-foc}
  \end{align}
  Furthermore, if the mechanism sets $c \geq \tfrac{1}{2}\max_{i,k} v_k^i$, then $\vect{a}$ is a pure-strategy Nash equilibrium if and only if (\ref{eq:foc}) and (\ref{eq:aggregate-foc}) are satisfied.
	
  Furthermore, for any vote totals $\vect{V}$, if $\vect{A}$ is any aggregate vote vector that satisfies (\ref{eq:aggregate-foc}), then for any $\vect{v}$ summing to $\vect{V}$ there exists a corresponding $\vect{a}$ summing to $\vect{A}$ satisfying (\ref{eq:foc}).
\end{lemma}

Summing over Equations \ref{eq:foc} and \ref{eq:aggregate-foc}, we obtain:
\begin{corollary} \label{cor:sum-indiv-votes-0}
	The sum of votes for each agent across alternatives is $0$ in equilibrium.
    The sum of aggregate votes across alternatives is $0$ in equilibrium: $\sum_{k=1}^m A_k = 0.$
\end{corollary}

An equilibrium is unique if Equation \ref{eq:aggregate-foc} has a unique solution $\vec{A}$.
In particular, we now prove Proposition \ref{prop:equilibrium-uniqueness}, which states that equilibria are unique when in the two alternative setting:

\begin{proposition} \label{prop:equilibrium-uniqueness}
	In the QTM on agent value profile $\vect{v}$ for $m = 2$ alternatives, if the mechanism chooses $c \geq \tfrac{1}{2} \max_{i,k} v_k^i$, then there is a unique pure-strategy Nash equilibrium. 
\end{proposition}

\begin{proof}[Proposition \ref{prop:equilibrium-uniqueness}]
	Fix all values $\vect{v}$ and set $c \geq \frac{1}{2} \max_{i, k} v_k^i$.
	We proceed by contradiction: assume two distinct equilibrium solutions exist $(p_1, p_2)$ and $(\hat{p}_1, \hat{p}_2)$.
	WLOG, assume $p_1 > \hat{p}_1$. 
	Since $V_1 > V_2$, by Corollary \ref{cor:two-alt-foc} we know $p_1 > \hat{p}_1 > 1/2$.
	It follows that $A_1 > \hat{A}_1$ since the function $\frac{1}{1 + e^{-2x}}$ is monotone increasing and $p_k = \frac{1}{1 + e^{-2A_k}}$ for $k = 1, 2$.
	
	By FOC (Corollary \ref{cor:two-alt-foc}), then, $A_1 - \hat{A}_1 = (p_1 p_2 - \hat{p}_1 \hat{p}_2) \frac{1}{2c} (V_1 - V_2) > 0$.
	Since $V_1 - V_2 > 0$, it follows that $p_1 p_2 - \hat{p}_1 \hat{p}_2 > 0$.
	In other words, $p_1 (1 - p_1) > \hat{p}_1 (1 - \hat{p}_1).$
	However, the function $x(1-x)$ is monotone decreasing for $x \geq 1/2$, so this is a contradiction since we assumed $p_1 > \hat{p}_1 > 1/2$.
	The result follows.
\end{proof}

\subsection{Two Alternatives} \label{app:two-alternatives}
This section contains proofs from Section \ref{sec:two-alternatives}.

We begin by proving Theorem \ref{thm:qtm-main-thm},.
We first observe an equivalent representation of the first-order conditions in the two-alternative case.
\begin{corollary} \label{cor:two-alt-foc}
  In the two-alternative case, in equilibrium we have $A_1 = \frac{p_1 p_2}{2c}(V_1 - V_2)$.
\end{corollary}
\begin{proof}
  By Lemma \ref{lem:foc},
  \begin{align*}
  	a_1^i &= \frac{p_1}{2c} (v_1^i - \expvi) \\
  	&= \frac{p_1}{2c} (v_1^i - p_1 v_1^i - (1 - p_1) v_2^i) \\
  	&= \frac{p_1 p_2}{2c} (v_1^i - v_2^i). %
  \end{align*}
  It follows that
  \begin{equation*}
  	A_1 =  \frac{p_1 p_2}{2c} (V_1 - V_2) . %
  \end{equation*}
\end{proof}

Our first result is that the Price of Anarchy in the two-alternative case always exceeds $1/2$.
This follows because, by Corollary \ref{cor:two-alt-foc} and $V_1 \geq V_2$, there are always at least as many votes for alternative $1$ as  for $2$.
\begin{proposition} \label{prop:poa-greater-half}
	In any pure-strategy equilibrium of the two-alternative QTM, $p_1 \geq 1/2$, with strict inequality if $V_1 > V_2$. 
	In particular, if $c \geq \frac{\max_{i, k} v_k^i}{2}$, then $\text{pPoA}(\vect{v}) > 1/2$. 
\end{proposition}
\begin{proof}
	$p_1 = \frac{e^{A_1}}{e^{A_1} + e^{A_2}}$ and $p_2 = \frac{e^{A_2}}{e^{A_1} + e^{A_2}}$. 
	By Lemma \ref{lem:foc}, $A_1 = \frac{p_1 p_2}{2c} (V_1 - V_2) \geq 0$.
	Also by Lemma \ref{lem:foc}, $A_2 = -A_1 \leq 0$.
	Thus $p_1 \geq p_2$, so $p_1 \geq 1/2$.
	If $V_1 > V_2$, then $A_1 > 0$ and $p_1 > p_2$.
	If $V_1 = V_2$ then $\text{pPoA} = 1$; otherwise, $\text{pPoA}(\vect{v}) \geq p_1 > 1/2$.
\end{proof}

We now consider the case where there is an additive gap between the aggregate values of the alternatives.
This is the main result behind both Theorem \ref{thm:qtm-main-thm} and Corollary \ref{cor:qtm-gap-poa}.
\begin{lemma} \label{lem:additive-gap-poa}
	In any pure-strategy equilibrium of the two-alternative QTM,
		\[ p_1 \geq 1 - \left(\frac{8c}{V_1 - V_2}\right)^{2/3} . \]
\end{lemma}
\begin{proof}
	Because votes sum to zero, $A_1 = -A_2$, and we can write $p_2 = e^{-A_1} / (e^{A_1} + e^{-A_1})$.
	Furthermore, by Corollary \ref{cor:two-alt-foc}, in equilibrium we have the following, and use the observation that $A_1 \geq 0$:
	\begin{align}
		A_1	&= \frac{p_1 p_2}{2c}(V_1 - V_2)  \nonumber \\
			&= \frac{V_1 - V_2}{2c \left(e^{A_1} + e^{-A_1}\right)^2}  \nonumber \\
			&\geq \frac{V_1 - V_2}{2c \left(2e^{A_1}\right)^2}  \nonumber \\
		\implies \quad A_1 e^{2A_1} &\geq \frac{V_1 - V_2}{8c}  \label{eq:poa-tighter} \\
		\implies \quad e^{A_1} &\geq \left(\frac{V_1 - V_2}{8c}\right)^{1/3} ,  \label{eq:poa-looser}
	\end{align}
	where the final inequality follows because $A_1 \leq e^{A_1}$, after which we cube-root both sides.

	Now, if $x := e^{A_1}$, then $p_2 = \frac{1/x}{x + 1/x} = \frac{1}{x^2 + 1} \leq \frac{1}{x^2}$.
	By (\ref{eq:poa-looser}), $p_2 \leq \left(\frac{8c}{V_1 - V_2}\right)^{2/3}$, proving the bound.
\end{proof}

\begin{remark}
	An asymptotically better convergence rate than in Theorem \ref{thm:qtm-main-thm} and Corollary \ref{cor:qtm-gap-poa} is achievable.
	By (\ref{eq:poa-tighter}), for all $\alpha > 0$, $e^{A_1 (2 + \alpha)} = \Omega\left(\frac{V_1 - V_2}{8c}\right)$.
	Following the same logic as Lemma \ref{lem:additive-gap-poa} and Corollary \ref{cor:qtm-gap-poa}, we obtain that for all $\epsilon > 0$, $\text{pPoA}(\vect{v}) \geq 1 - O\left(\left(\frac{1}{G}\right)^{1-\epsilon}\right)$.
	Following the same logic as in the proof of Theorem \ref{thm:qtm-main-thm}, it will turn out that the factor $\left(\frac{2}{T}\right)^{2/5}$ can be improved to $O\left(\frac{2}{T}\right)^{\frac{1}{2} - \epsilon}$ for any $\epsilon > 0$, at the cost of a correspondingly large constant factor.
\end{remark}

\begin{proof}[Theorem \ref{thm:qtm-main-thm}]
	Let $c = \tfrac{1}{2} \max_{i,k} v_k^i$.
	In this case, by Proposition \ref{prop:equilibrium-existence}, a pure-strategy equilibrium is guaranteed to exist.
	Let $Y = (8c)^{2/5}(2V_1)^{3/5}$, a magically chosen quantity that balances the following cases.

	Case $V_1 - V_2 \leq Y$: here, we observe that the Price of Anarchy is
	\begin{align*}
		\frac{p_1 V_1 + p_2 V_2}{V_1}
		&\geq \frac{V_1 + V_2}{2V_1}  & \text{using $p_1 \geq \tfrac{1}{2}$}  \\
		&\geq \frac{2V_1 - Y}{2V_1}  \\
		&=    1 - \frac{Y}{2V_1}  \\
		&=    1 - \left(\frac{4c}{V_1}\right)^{2/5} .
	\end{align*}

	Case $V_1 - V_2 \geq Y$: here, we apply Lemma \ref{lem:additive-gap-poa} to obtain that the Price of Anarchy exceeds
	\begin{align*}
		p_1 &\geq 1 - \left(\frac{8c}{Y}\right)^{2/3}  \\
		    &=    1 - \left(\frac{4c}{V_1}\right)^{2/5} .
	\end{align*}
	In both cases, we obtain $1 - \left(\frac{2}{T}\right)^{2/5}$.
\end{proof}

\paragraph{Total payments.}
Define the \emph{revenue} of the QTM to be the total payments of the agents before redistribution, $c \sum_i \sum_k (a_k^i)^2$.
It is interesting, and will be useful later, to understand the relationship of revenue to the type profile $\vect{v}$.
As with the pPoA analysis, we can analyze revenue using the first-order conditions.
\begin{proposition} \label{prop:revenue}
	In the QTM with two alternatives and type profile $\vect{v}$, with $V_1 > V_2$, define the \emph{disagreement} to be $D := \frac{\sum_i (v_1^i - v_2^i)^2}{(V_1 - V_2)^2}$.
	Then in any pure-strategy equilibrium,\footnote{Here $\Theta(\cdot)$ is with respect to $V_1 - V_2$ growing, or at least bounded below. The proof includes fully explicit nonasymptotic bounds.}
		\[ \text{revenue} = \Theta\left( c ~ D \left(\ln \tfrac{V_1 - V_2}{c} \right)^2 \right) . \]
\end{proposition}

\begin{proof}
	We heavily use that, in equilibrium, each individual's votes sum to zero.
  
	By the first-order conditions , $A_1 = \frac{p_1 p_2}{2c}(V_1 - V_2)$, or rearranging,
	\begin{align}
	  p_1 p_2 &= \frac{2 c A_1}{V_1 - V_2} . \label{eqn:p1p2-foc}
	\end{align}
	Also by , $a_1^i = \frac{p_1 p_2}{2c}(v_1^i - v_2^i)$.
	So the revenue is
	\begin{align}
	  \sum_i \left[ c (a_1^i)^2 + c (a_2^i)^2 \right]
	  &= 2c \sum_i (a_1^i)^2  \nonumber \\
	  &= 2c \left(\frac{p_1 p_2}{2c}\right)^2 \sum_i (v_1^i - v_2^i)^2  \label{eq:budget-helper-1} \\
	  &= 2c \left(\frac{A_1}{V_1 - V_2}\right)^2 \sum_i (v_1^i - v_2^i)^2  \nonumber \\
	  &= 2c (A_1)^2 \frac{\sum_i (v_1^i - v_2^i)^2}{(V_1 - V_2)^2} . \label{eq:budget-helper-2}
	\end{align}
	It only remains to bound $(A_1)^2$.
	\begin{align}
	  p_1 p_2 &= \frac{e^{A_1} e^{-A_1}}{(e^{A_1} + e^{-A_1})^2}  \nonumber \\
			  &= \frac{1}{(e^{A_1} + e^{-A_1})^2}  \nonumber \\
			  &\geq \frac{1}{(2e^{A_1})^2}  \label{eq:revenue-ineq} \\
			  &= \frac{1}{4e^{2A_1}} .  \nonumber
	\end{align}
	So
	\begin{align*}
	  A_1 &=    \frac{p_1 p_2}{2c}(V_1 - V_2)  \\
		  &\geq \frac{V_1 - V_2}{8c e^{2A_1}} .
	\end{align*}
	Rearranging and taking the logarithm of both sides,
	\begin{align*}
	  \ln\left(\frac{V_1 - V_2}{8c}\right)
	  &\leq 2 A_1 + \ln(A_1)  \\
	  &\leq 3 A_1 .
	\end{align*}
	It follows that~%
	\begin{align*}
	  A_1^2 &\geq \left(\max\left\{ 0 ~,~ \frac{1}{3}\ln \tfrac{V_1 - V_2}{8c} \right\} \right)^2 .
	\end{align*}
	On the other hand, we can replace Inequality \ref{eq:revenue-ineq} as follows:
	\begin{align*}
	  p_1 p_2 &= \frac{1}{(e^{A_1} + e^{-A_1})^2}  \\
	  &\leq \frac{1}{(e^{A_1})^2} ,
	\end{align*}
	and continuing with almost the same analysis, using $A_1 \geq 0$ we obtain
	\begin{align*}
	  \ln\left(\frac{V_1 - V_2}{2c}\right) &\geq 2A_1 .
	\end{align*}
	We have obtained
	  \[ \max \left\{ 0 ~,~ \frac{1}{3}\ln\left(\frac{V_1 - V_2}{8c}\right) \right\} \leq A_1 \leq \frac{1}{2}\ln\left(\frac{V_1 - V_2}{2c}\right) , \]
	so
	  \[ \frac{2c}{9} \frac{\sum_i (v_1^i - v_2^i)^2}{(V_1 - V_2)^2} \left(\max \left\{ 0 ~,~ \ln \tfrac{V_1 - V_2}{8c} \right\} \right)^2 \leq \text{revenue} \leq \frac{c}{2} \frac{\sum_i (v_1^i - v_2^i)^2}{(V_1 - V_2)^2} \left(\ln \tfrac{V_1 - V_2}{2c}\right)^2 . \]
  \end{proof}

In other words, the revenue is tightly controlled by the disagreement $D$ of the strategy profile.
For intuition, suppose the strengths of the agent preferences are identical with $|v_1^i - v_2^i| = 1$.
Then the numerator of $D$ is $n$.
Let $n_1$ be the number of agents with $v_1^i - v_2^i = 1$ and let $n_2$ be the number with $v_1^i - v_2^i = -1$.
Observe that $V_1 - V_2 = n_1 - n_2$, the number of agents who prefer alternative $1$ minus the number who prefer $2$.
\begin{itemize}
	\item If all agents agree, i.e. $n_1 = n$, then $V_1 - V_2 = n$ and the denominator is $n^2$; revenue converges to zero quickly.
	\item If $n_1 - n_2 = \Theta(\sqrt{n})$, then $D = \Theta(1)$, and revenue grows polylogarithmically in $n$.
	\item For $n_1 - n_2 \ll \sqrt{n}$, $D = \omega(1)$ and revenue grows more rapidly.
\end{itemize}
The disagreement $D$ has a natural statistical interpretation.
If we pick an agent $i$ uniformly at random and let $X = v_1^i - v_2^i$, then $D = \frac{1}{n} \left(\frac{\sigma^2 + \mu^2}{\mu^2}\right)$, where $\mu$ and $\sigma$ are the mean and standard deviation of $X$.

\subsection{QTM with m alternatives} \label{appendix:qtm-m-alternatives}

We briefly touch on the case where there are more than two alternatives.
\cite{Eguia19,Eguia23} give asymptotic results showing that in their settings, with a fixed distribution of bounded agents, as $n \to \infty$ the welfare of the QTM tends to optimality.
Here we can give a weaker but nonasymptotic bound of $\geq \frac{1}{m}$.
The proof is not immediate and involves dividing the alternatives into a group with ``high'' welfare (higher than the mechanism's expectation) and the remainder, an idea from \cite{Eguia23} used in a different way.

\begin{proof}[Proof of \ref{prop:m-poa}]
	By Proposition \ref{prop:equilibrium-existence}, for this regime of $c$, a pure-strategy equilibrium exists.
	Assume WLOG that $V_1 \geq \cdots \geq V_m$.
	Recall that $\vect{p}$ is the mechanism's output, i.e. distribution on alternatives, and the welfare is $\sum_{\ell} p_{\ell} V_{\ell}$.
	Let
		\[ \bar{V} \coloneq \sum_{\ell} p_{\ell} V_{\ell} , \]
	the welfare of the mechanism.
	We seek to prove $\frac{\bar{V}}{V_1} \geq \frac{1}{m}$.
	We will lean on the first-order conditions for equilibrium (Lemma \ref{lem:foc}): reproducing Equation (\ref{eq:aggregate-foc}), for each alternative $k$,
	\begin{equation} \label{eq:foc-restated-m}
		A_k  = \frac{p_k}{2c} \left(V_k - \bar{V} \right) .
	\end{equation}

	We will consider the following subsets of the alternatives:
	\begin{itemize}
		\item $S = \{k : V_k \geq \bar{V} \}$.
		\item $T = \{k : V_k < \bar{V} \}$.
		\item $X = \{1\} \cup T$.
		\item $Y = S \setminus \{1\}$.
	\end{itemize}
	Observe that $1 \in S$ and $m \in T$ by Equation \ref{eq:foc-restated-m}.
	Let $p(S) = \sum_{k \in S} p_k$, the probability of an outcome in $S$, and similarly for $p(T),p(X),p(Y)$.
	Let $w(S) = \frac{1}{p(S)} \sum_{k \in S} p_k V_k$, the expected welfare conditioned on picking an outcome in $S$, and similarly for $w(T),w(X),w(Y)$.
	We will deal with the case $|Y| = 0$ separately.

	We first claim $\frac{w(X)}{p(X)} \geq \frac{1}{m}$, or in other words, the expected welfare conditioned on being in $X$ exceeds $\frac{V_1}{m}$.
	By (\ref{eq:foc-restated-m}), $A_1 > 0$ and $A_k < 0$ for all $k \in T$.
	Therefore, $e^{A_1} > e^{A_k}$ for all $k \in T$.
	Therefore,
	\begin{align*}
		w(X)	&= 		\frac{1}{p(X)} \left(p_1 V_1 + \sum_{k \in T} p_k V_k\right)  \\
				&\geq	\frac{p_1}{p(X)} V_1  \\
				&=		\frac{e^{A_1}}{e^{A_1} + \sum_{k \in T} e^{A_k}} V_1  \\
				&\geq	\frac{e^{A_1}}{e^{A_1} (|T| + 1)} V_1 \\
				&=		\frac{V_1}{|T|+1}  \\
				&\geq	\frac{V_1}{m} .
	\end{align*}
	Now, if $|Y| = 0$, then $|X| = m$ and $w(X) = \text{Welfare}$, and we are done.
	Otherwise, we claim $w(Y) \geq \bar{V}$, because for all $k \in Y$, since $k \in S$, we have $V_k \geq \bar{V}$.
	But since $Y$ and $Z$ partition the alternatives,
	\begin{align*}
		\bar{V}	&=		p(Y) w(Y) + p(Z) w(Z)  \\
				&\geq 	p(Y) \frac{V_1}{m} + p(Z) \bar{V}  \\
		\implies \bar{V} \geq \frac{V_1}{m} .
	\end{align*}
\end{proof}

It is not clear if this bound can be improved non-asymptotically.
The analysis of the many-alternatives case is very complex (see \cite{Eguia19,Eguia23}).
As far as we currently know, it may be possible to have instances of a large ``spread'' $\frac{V_1}{\max_{i,k} v_k^i}$, perhaps exponentially large in $m$, with pure-strategy Price of Anarchy bounded away from one, or perhaps as small as $O(1/m)$.

\section{SQUAP Price of Anarchy} \label{app:predictions}

\begin{proof}[Proof of Theorem \ref{thm:general-aggr-general-m}]
	By Proposition \ref{prop:synthetic-fixed-ppoa-bb-m}, we have $\min_{\vect{p} \in \text{pPNE}(\vect{v}, \D)} \frac{\sum_k p_k \hat{W}_k}{\hat{W}_1} \geq C(\vect{v})$, and we know $\max_k |\hat{B_k} - B_k| \leq \alpha  \max_{i,k} v_k^i$ by Lemma \ref{lem:aggr-bhat}. 
	It follows that, for $\vect{p} \in \arg\min_{\vect{p} \in \text{pPNE}(\vect{v}, \D)} \frac{\sum_k p_k W_k}{W_1}$,
	\begin{align*}
		\text{pPoA}(\vect{v},\D) &= \frac{\sum_k p_k W_k}{W_1} \\
		&\geq  \frac{\sum_k p_k \hat{W}_k - \alpha \max_{i,k} v_k^i}{\hat{W}_1 + \alpha \max_{i,k} v_k^i} \\
		&\geq \frac{\sum_k p_k \hat{W}_k - \alpha \max_{i,k} v_k^i}{\hat{B}_1 + (1 + \alpha) V_1} \\
		&\geq \frac{\sum_k p_k \hat{W}_k - \alpha \max_{i,k} v_k^i}{(1 + \alpha)\hat{W}_1} \\\
		&\geq \frac{1}{1 + \alpha} \left( C(\vect{v}) - \frac{\alpha}{T} \right),
	\end{align*}
	where the second inequality follows from the fact that $\alpha \max_{i,k} v_k^i \leq \alpha V_1$.
\end{proof}

\section{Prediction Markets} \label{app:prediction-market}

We prove the $(\alpha, x)$-robustness of the importance-weighted decision market.
\begin{proof}[Proof of Lemma~\ref{lem:market-deviation}]
	Observe that under the efficient market assumption, the decision market is equivalent to the following: run the entire market, where the last participant predicts $\vect{B^*}$.
	Then, give the last participant one more opportunity to participate; denote their final predictions by $\vect{\hat{B}}$.
	This is equivalent because of the telescoping sum of the market scoring rule, so that the two prediction opportunities have the same net payoff as one opportunity with predicting $\vect{\hat{B}}$.

	For any participant except the final one, under the efficient market assumption, their strategy does not affect the final predictions $\vect{\hat{B}}$.
	One way to see this is by the above equivalence, since the market will end at $\vect{B}^*$ regardless, followed by the last participant's possible manipulation.

	So we only need to analyze the incentives of the final participant $i$.
	By the decomposition above, we only need to consider the incentives for the final prediction opportunity (since the previous one is fixed to $\vect{B^*}$).
	The expected utility for $\vect{\hat{B}}$ is
	\begin{align*}
		S(\vect{\hat{B}};\vect{B^*}) - S(\vect{B^*};\vect{B^*})
		&= S(\vect{\hat{B}};\vect{B^*}) \\
		&= -\frac{1}{\beta}\sum_k (\hat{B}_k - B_k^*)^2.   & \text{(via Equation \ref{eq:expected-score})}
	\end{align*}
	By the definition of $x$-best-responding in the market,
	$x > \frac{1}{\beta} \sum_k (\hat{B}_k - B_k^*)^2$.
	So for all $k$, $(\hat{B}_k - B_k^*)^2 \leq \beta x = \epsilon x^2$, or $|\hat{B}_k - B_k^*| \leq \sqrt{\epsilon} x$.
\end{proof}

We prove that, in strategy profiles with high ``disagreement'', SQUAP with prediction markets is budget balanced in expectation.
\begin{proof}[Proof of Corollary \ref{cor:prediction-market-bb}]
  We prove that the mechanism can choose liquidity parameters $\beta$ so that the expected revenue is larger than the amount spent on the decision market stage in expectation.
  In practice, one would like to scale the prediction-market payments by the largest possible amount so that budget balance is achieved.
  However, we cannot rule out manipulation for reasons similar to the redistribution, so for analysis we assume that the mechanism designer commits to $\beta$ up front.
  By the telescoping nature of the proper scoring rule, the mechanism's total spend on the decision market stage in expectation is $S(\vect{\hat{B}};\vect{B^*}) - S(\vect{B^0},\vect{B^*}) \leq S(\vect{B^*};\vect{B^*}) - S(\vect{B^0};\vect{B^*}) = \frac{1}{\beta}\sum_k (B_k^0 - B_k^*)^2$.
  With at-least-constant disagreement, revenue is at least constant (by Proposition \ref{prop:revenue}), so $\beta = \Omega(1)$ is possible with ex ante budget-balance.
  By Corollary \ref{cor:market-poa}, in this case and with $T \to \infty$, the Price of Anarchy converges to $1$.
\end{proof}

\end{document}